\documentclass[acmsmall]{ec20acm}

\usepackage{booktabs} 
\usepackage[ruled]{algorithm2e} 

\SetAlFnt{\small}
\SetAlCapFnt{\small}
\SetAlCapNameFnt{\small}
\SetAlCapHSkip{0pt}
\IncMargin{-\parindent}

\setcitestyle{authoryear}

\usepackage{hyperref}
\hypersetup{colorlinks,
	linkcolor=blue,
	citecolor=blue,
	urlcolor=magenta,
	linktocpage,
	plainpages=false}
\usepackage{amsthm}
\usepackage{latexsym}
\usepackage{amsmath}
\usepackage{amssymb}
\usepackage{algcompatible}
\usepackage{color}
\usepackage{tikz}
\usepackage{multirow}
\usepackage{enumerate}
\usepackage{subfigure}
\usepackage{threeparttable}
\usepackage[title]{appendix}

\makeatletter
\newtheorem*{rep@theorem}{\rep@title}
\newcommand{\newreptheorem}[2]{%
\newenvironment{rep#1}[1]{%
 \def\rep@title{#2 \ref{##1}}%
 \begin{rep@theorem}}%
 {\end{rep@theorem}}}
\makeatother

\newcommand{\old}[1]{}

\newtheorem{theorem}{Theorem}
\newreptheorem{theorem}{Theorem}
\newtheorem{definition}{Definition}[section]
\newtheorem{lemma}[definition]{Lemma}
\newreptheorem{lemma}{Lemma}
\newreptheorem{informal_lemma}{Informal Lemma}
\newtheorem{corollary}[definition]{Corollary}

\newtheorem{proposition}[definition]{Proposition}
\newreptheorem{proposition}{Proposition}
\newtheorem{mechanism}{Mechanism}
\newreptheorem{mechanism}{Mechanism}

\newcommand{\R}{\mathbb{R}}
\newcommand{\E}{\mathop{\mathbb{E}}}

\newcommand{\soc}{\mathsf{sc}}
\newcommand{\mc}{\mathsf{mc}}
\newcommand{\obj}{\mathsf{obj}}

\newcommand{\vx}{\boldsymbol{x}}
\newcommand{\vy}{\boldsymbol{y}}

\AtBeginDocument{%
  \providecommand\BibTeX{{%
    \normalfont B\kern-0.5em{\scshape i\kern-0.25em b}\kern-0.8em\TeX}}}

\copyrightyear{2020}
\acmYear{2020}
\setcopyright{acmcopyright}
\acmConference[EC '20] {Proceedings of the 21st ACM conference on Economics and Computation}{July 13--17, 2020}{Virtual Event, Hungary}
\acmBooktitle{Proceedings of the 21st ACM conference on Economics and Computation (EC '20), July 13--17, 2020, Virtual Event, Hungary}
\acmPrice{15.00}
\acmDOI{10.1145/3391403.3399471}
\acmISBN{978-1-4503-7975-5/20/07}

\begin{document}

\title{Characterization of Group-Strategyproof Mechanisms for Facility Location in Strictly Convex Space}


\author{Pingzhong Tang}
\affiliation{%
  \institution{IIIS, Tsinghua University}}
\email{kenshinping@gmail.com}

\author{Dingli Yu}
\affiliation{%
  \institution{Computer Science Department, Princeton University}}
\email{leo.dingliyu@gmail.com}

\author{Shengyu Zhao}
\affiliation{%
  \institution{IIIS, Tsinghua University}}
\email{zsyzzsoft@gmail.com}


\begin{abstract}
We characterize the class of group-strategyproof mechanisms for the single facility location game in any unconstrained strictly convex space. A mechanism is \emph{group-strategyproof}, if no group of agents can misreport so that all its members are \emph{strictly} better off. A strictly convex space is a normed vector space where $\|x+y\|<2$ holds for any pair of different unit vectors $x \neq y$, e.g., any $L_p$ space with $p\in (1,\infty)$.

We show that any deterministic, unanimous, group-strategyproof mechanism must be dictatorial, and that any randomized, unanimous, translation-invariant, group-strategyproof mechanism must be \emph{2-dictatorial}. Here a randomized mechanism is 2-dictatorial if the lottery output of the mechanism must be distributed on the line segment between two dictators' inputs. A mechanism is translation-invariant if the output of the mechanism follows the same translation of the input.

Our characterization directly implies that any (randomized) translation-invariant approximation algorithm satisfying the group-strategyproofness property has a lower bound of $2$-approximation for maximum cost (whenever $n \geq 3$), and $n/2 - 1$ for social cost. We also find an algorithm that $2$-approximates the maximum cost and $n/2$-approximates the social cost, proving the bounds to be (almost) tight.




\end{abstract}




\maketitle

\clearpage

\section{Introduction}

In a single facility location game of $n$ agents, every agent reports a location, and the mechanism chooses a facility location (i.e., an alternative). The cost of each agent is the distance between the facility location and her true location. A mechanism is strategyproof, if no one can be better off by misreporting, i.e., the outcome of the mechanism cannot be closer even if she reports a fake location. A mechanism is group-strategyproof, if no group of agents can jointly misreport their inputs to the algorithm so that \emph{all} members are \emph{strictly} better off.

Characterization of truthful mechanisms for voting problems has received a great amount of attention in the past few decades. Seminal results include \citet{Gibbard1977Voting} that shows any strategyproof mechanism that depends on individual strong orderings must be a probability mixture of \emph{unilateral} (i.e., only one agent can affect the outcome) or \emph{duple} (restricting the deterministic outcome between a fixed pair of alternatives) schemes.

Unfortunately, such characterization is generally hard to achieve in the restricted domains such as single-peaked preferences (i.e., one-dimensional facility location game) and especially in the higher-dimensional settings, however, they also provide possibilities for positive results. An important breakthrough in this literature was made by \citet{Moulin1980} --- a complete characterization of deterministic strategyproof mechanisms over one-dimensional single-peaked preferences --- known as the median voter schemes. \citet{Border1983Straightforward} extend Moulin's result to Euclidean space and show that it induces to median voter schemes in each dimension separately. \citet{Barbera1993Median} further show that this separability generalizes to any $L_1$-norm, known as generalized median voter schemes. Some more restricted domains are also studied, e.g., a compact set of the Euclidean space~\cite{Barbera1998Compact}, while typically showing that all those mechanisms behave like generalized median voter schemes. However, there is no result about randomized mechanisms in such settings.

An approximation perspective of this problem was first introduced by \citet{Procaccia2009Approximate}. They study approximately optimal strategyproof mechanisms for facility games under social cost (the sum of individual costs) and maximum cost (the maximum of individual costs) objectives while taking care of randomized mechanisms, but focus on the one-dimensional setting. They propose an interesting randomized mechanism for the maximum cost objective, which $3/2$-approximates the optimum and proves to be the best, while also ensuring group-strategyproofness. Many follow-up works are concerned with the approximation bounds of facility games on a line or a network, including many facilities and non-linear costs (see e.g., \cite{Alon2010Networks,Lu2010Two,Escoffier2011Many,Fotakis2013Concave,Feldman2013Tree,Fotakis2014Det,Cai2016Facility,Feldman2016Facility,Fong2018Facility,Golowich2018Deep}). However, to the best of our knowledge, none of the previous works studies randomized mechanisms in any multi-dimensional space --- this is substantially different from the discrete network cases, but no characterization is established, nor any approximation bound. In fact, \cite{Border1983Straightforward} implies that some of those strategyproof mechanisms can indeed be generalized to the multi-dimensional domain and preserve the property of strategyproofness. Unfortunately, they are no longer group-strategyproof, which is the focus of this paper.

We study a fundamental aspect of this problem: single facility location, in an unconstrained multi-dimensional strictly convex space. For example, Euclidean distance, or generally any $L_p$-norm ($p \in (1, \infty)$), after an arbitrary affine transformation, is strictly convex. The motivation of studying strictly convex norms is not only that they are commonly used in the related literature, and moreover, agents in a strictly convex space are more likely to manipulate and behave unlike separable preferences established in the previous characterizations~\cite{Border1983Straightforward,Barbera1993Median}.

A special case of our setting --- deterministic group-strategyproof mechanisms in the Euclidean space --- has long been investigated. \citet{Bordes2011Euclidean}, originally presented as~\citep{Bordes1990Euclidean} 20 years earlier, show that any deterministic group-strategyproof mechanisms in the unconstrained Euclidean space must be dictatorial. Their proof aims to characterize the option sets and utilizes many properties specified to the Euclidean distance, which seems hardly extendable to general strictly convex norms. Our deterministic part can be regarded as a direct generalization of this result but proved in a different and possibly cleaner way.
\citet{Sui2015Multidim} studies this setting in a practical manner, showing that generalized median mechanisms are not group-strategyproof in the unconstrained Euclidean space while the incentive of group misreport is unbounded. Yet, the approximation bound or even possibility of group-strategyproofness with respect to any other norm is still unknown, let alone randomized mechanisms. Disregarding algorithmic randomness partially circumvents the computational complexity and usually comes with more restrictive results. In this paper we show that randomized mechanisms can indeed implement more.

\subsection{Our Results}



We start by considering the deterministic case in Section \ref{section:deterministic}, and our main theorem in this part states as follow.

\begin{reptheorem}{DictatorialTheorem}
If $f$ is deterministic, unanimous, and group-strategyproof, then $f$ is dictatorial.
\end{reptheorem}

This characterization is actually complete, i.e., conversely, any dictatorial mechanism must be deterministic, unanimous, and group-strategyproof.
Here a mechanism is unanimous if all agents report the same point, the mechanism must choose that point as well. In fact, a mechanism that is non-unanimous can be unbounded in terms of approximation under both social cost and maximum cost objectives. Our approach in this part is a fundamental building block that introduces some basic lemmas and intuitions for the characterization of randomized mechanisms.


For randomized mechanisms, we introduce a condition called translational invariance, which says that if we apply a translation to the inputs, the mechanism must output a location that is the result of the same translation to the original output. This condition is quite natural in the facility location domain (also known as shift invariance in~\cite{Feigenbaum2016LpNorm} and position invariance in~\cite{Filos2017DoublePeaked}), and we indeed find some strange mechanisms (e.g., Mechanism \ref{Separate2Dictator}) that are not translation-invariant and related to some constant.
Our main theorem states as follow.

\begin{reptheorem}{2DictatorialTheorem}
If $f$ is unanimous, translation-invariant, and group-strategyproof, then $f$ is 2-dictatorial.
\end{reptheorem}

A mechanism is \emph{2-dictatorial}, if the support of the output always lies on the segment between two dictators' inputs (the two dictators are called 2-dictators). Theorem \ref{2DictatorialTheorem} also indicates an impossibility result with respect to the anonymity (Corollary \ref{2DictatorialCorollary}). However, on the positive side, it forms a game on the line between the 2-dictators, and we present a non-trivial design in Mechanism \ref{RandMed}.

All those constraints are essential in guaranteeing the clean characterizations. Without the unanimity, there exist some ``constant-related'' mechanisms, e.g., randomly choosing either the dictator's location or a constant displacement to that location is translation-invariant and group-strategyproof. Without the translation invariance, we also construct a constant-related but a bit more complicated mechanism (Mechanism \ref{Separate2Dictator}) that is unanimous and group-strategyproof but not 2-dictatorial. Without the group-strategyproofness, there would be a large possibility of choices as discussed in Section \ref{sec:discussion}.




Based on the characterizations above, we obtain lower bounds of approximately optimal mechanisms under both maximum and social cost objectives, summarized in Table \ref{table:summary}. The upper bounds for randomized mechanisms are guaranteed by Mechanism \ref{RandMed}. All bounds are tight, except for the little gap of randomized mechanisms for the social cost objective.

\begin{table}[t]
\centering
\begin{threeparttable}
\begin{tabular}{|c|c|c|c|}
\hline
\multicolumn{2}{|c|}{} & Deterministic & Randomized \\
\hline
\multirow{2}*{maximum cost} & $n = 2$ & $[2, 2]$ & $[3/2\tnote{a}~, 3/2]$ \\ \cline{2-4}
    & $n \geq 3$ & $[2, 2]$ & $[2\tnote{b}~, 2]$ \\
\hline
social cost & $n \geq 2$ & $[n - 1, n - 1]$ & $[n/2 - 1\tnote{b}~, n/2]$ \\
\hline
\end{tabular}
\caption{Summary of the approximation bounds of group-strategyproof mechanisms in strictly convex space. $n$ represents the number of agents.}
\label{table:summary}
\begin{tablenotes}
\item[a] Proved by \citet{Procaccia2009Approximate} in the one-dimensional case.
\item[b] Requires translational invariance.
\end{tablenotes}
\end{threeparttable}
\end{table}

\subsection{Our Techniques}

We develop a series of technical tools to facilitate our characterizations. First, when some results are obtained in the small-scale cases, we use \emph{scale reduction} to generalize them over all $n$. Second, we apply \emph{output space reduction} to rule out some locations or randomized distributions that are not preferred by any agent; typically, the output should be a convex combination over some agents' inputs. Third, we establish a fundamental property called \emph{cost continuity}, which is essential in our proofs to generalize local properties over the whole space. Those techniques are detailed as follows.

\subsubsection{Scale Reduction}

For both characterizations, we begin with some small-scale cases ($n = 2$ for deterministic mechanisms and $n = 3$ for randomized mechanisms), and then reduce general $n$-agent games into those base cases. The intuition of the scale reduction is to bind some of the agents together to form a temporary coalition, while the reduced game still preserves the required properties (e.g., group-strategyproofness).

Our scale reduction in Theorem \ref{DictatorialTheorem} is similarly used by Bordes et al.'s characterization~\cite{Bordes2011Euclidean}. Essentially, assuming that the statement holds for the base case ($n = 2$), we first divide the agents into two groups and bind them respectively, one of which will be the dictator in the reduced $2$-agent game. Note that the dictator group forms some ``partial unanimity'' property, i.e., the unanimity with respect a subset of agents. Then we can fix the location of an agent in the non-dictator group, and the reduced $(n - 1)$-agent game is still unanimous and group-strategyproof and thus dictatorial by induction. It remains to show that all reduced games have the same dictator (see the proof of Theorem \ref{DictatorialTheorem} for details).

In Theorem \ref{2DictatorialTheorem}, our scale reduction approach is substantially different. Using the base case ($n = 3$), we can similarly divide the agents into three groups and two of them will be the 2-dictator groups, however, we cannot simply fix the location of an agent since then the reduced game is not necessarily translation-invariant. Here we use a different strategy: we fix an agent in the non-dictator group to some other agent (i.e., bind them together) and reduce the game into $n - 1$ agents, but then we still need to show that the 2-dictatorship of this special game generalizes to the whole space. This is still very complicated, and here we only point out our key observations: we can utilize the ``partial unanimity'' with respect to the 2-dictator groups, as well as the ``partial unanimity'' with respect to the fixed agent together with the 2-dictators in the special game. Note that the intersection of those ``partial unanimity'' groups is exactly the 2-dictators, and those ``partial unanimity'' can be further applied for output space reduction. See the proof of Theorem \ref{2DictatorialTheorem} for details.


\subsubsection{Output Space Reduction}

If we can find a location that is strictly preferred by everyone, then all agents can collaborate to misreport that location, and thus group-strategyproofness is violated by unanimity. This is the basic idea for output space reduction.

In the deterministic base case ($n = 2$), as long as the output does not lie on the line segment between the two agents, we can always find such location by the strict convexity. This is exactly what Lemma \ref{SegmentLemma} establishes.

For randomized mechanisms, we find the \emph{centroid} (i.e., the expected location of a distribution) very helpful to reduce the space of randomized distributions, as selecting the centroid of the output does not make anyone worse off. This observation gives us a fundamental understanding of randomized group-strategyproof mechanisms, as our Lemma \ref{CentroidLemma} states. We rephrase this lemma below in an informal way.

\begin{lemma}[Informal version of Lemma~\ref{CentroidLemma}]
If the mechanism is unanimous and group-strategyproof, then for any input preferences, the output of the mechanism is either deterministic, or a probability mixture over the line segment between some two agents' inputs.
\end{lemma}

\subsubsection{Cost Continuity}
Cost continuity is a fundamental local property of the output, which states informally as follow.
\begin{lemma}[Informal version of Lemma~\ref{ContinuousLemma}, Cost continuity]
If the mechanism is strategyproof, then for any agent $i$, assuming that the inputs of all other agents are fixed, then the expected distance between the output and agent $i$'s location is a continuous function, and moreover, the distance change cannot be larger than agent $i$'s movement (i.e., the function is 1-Lipschitz continuous).
\end{lemma}

Some previous work establishes the output continuity (stronger than cost continuity) for deterministic mechanisms~\cite{Peters1993Continuity}, but note that the continuity with respect to randomized distributions is essentially different, and here we care about the cost function.

Cost continuity itself might not seems strong at first glance, but it turns out to be essential for later proofs. By moving in small steps, we utilize the cost continuity to further prove other local properties and successfully generalize them to the whole space. Lemma \ref{InfiniteRatioLemma} especially makes the case.

\begin{lemma}[Informal version of Lemma~\ref{InfiniteRatioLemma}]
When $n = 3$, if the mechanism is unanimous and group-strategyproof, then there exists an input profile such that the output is relatively very close to an agent and isolated from another agent.
\end{lemma}

The formal version of this lemma is a pivot step of our final result. Cost continuity is used to keep the output distribution staying on the base of an isosceles triangle, while slowly reshaping the triangle to approach nearly a doubled line (and then the apex of the triangle will be the isolated agent). If the output ``jumps'' from the base to some leg, then its distance to some agent must not be continuous and we further show the contradiction by moving some two agents simultaneously. For more details, please refer to the formal statement and the proof.

Finally, we would like to remark that the cost continuity also holds for strategyproof mechanisms, which shows its potential for future work.

\section{Preliminaries}

We consider the single facility location game with $n$  ($n \geq 2$) agents $N = \{1, \ldots, n\}$. All agents are located in a $d$-dimensional strictly convex space $\R_d$ ($d \geq 2$), i.e., a normed vector space where $\forall x, y \in \R_d$, $x \neq y$ and $\|x\| = \|y\|$ together imply $\|x\| + \|y\| > \|x + y\|$, or equivalently, the triangle inequality $\|x + y\| \leq \|x\| + \|y\|$ holds with equality if and only if $x, y$ are in the same direction. Real number field is considered for simplicity. Let $x_i\in \R_d$ be the location of agent $i$'s location. A location profile is a vector of agents' location, $\vx = (x_1,\ldots,x_n)$.

A deterministic mechanism is a function $f:\R_d^n\to \R_d$, mapping a location profile to the location of facility. Given the location of facility $f(\vx) = y\in \R_d$, the cost of agent $i$ is the distance between $x_i$ and $y$, i.e., $\|x_i-y\|$. We say $f(\vx)$ is the output of $\vx$.

A randomized mechanism is a function mapping a location profile $\vx \in \R_d^n$ to a probability distribution over $\R_d$. Given a probability distribution of facility $f(\vx) = P$, the cost of agent $i$ is the expected distance between $x_i$ and the facility, i.e., \[\|x_i - P\| \triangleq \E_{y\sim P} \|x_i - y\|.\]

In what follows, we formally define several constraints and/or properties of a mechanism, which will be thoroughly discussed in the paper.

\begin{definition}[Strategyproofness]
A mechanism is \emph{strategyproof} if and only if no agent can gain from misreporting the location, that is, for all $\vx\in \R_d^n$, for all $i\in N$, and for all $x_i'\in \R_d$,
\[\|f(\vx) - x_i\|\leq \|f(x_i',\vx_{-i}) - x_i\|,\]
where $\vx_{-i} = (x_1,\ldots,x_{i-1},x_{i+1},\ldots,x_n)$ is the location profile without $x_i$.
\end{definition}

\begin{definition}[Group-strategyproofness]
A mechanism is \emph{group-strategyproof} if and only if for all $S\subseteq N$, there is no $x_S'\in \R_d^{|S|}$ such that all agents in $S$ can gain from misreporting, that is, for all $\vx \in \R_d^n$, for all $S\subseteq N$, for all $x_S'\in \R_d^{|S|}$, there exists $i\in S$ such that
\[\|f(\vx) - x_i\| \leq \|f(x_S', \vx_{-S}) - x_i\|,\]
where $\vx_{-S}$ is the location profile without agents in $S$.
\end{definition}

\begin{definition}[Unanimity]
A mechanism $f$ is \emph{unanimous} if and only if  \[x_1=\cdots=x_n=x \implies f(\vx) = x.\] That is, if all agents report the same point, the mechanism must choose that point as well.
\end{definition}

\begin{definition}[Translational invariance]
A mechanism $f$ is \emph{translation-invariant} if and only if \[\forall \vx\in \R_d^n, \forall a \in \R_d, f(\vx) + a = f(\vx + a), \]
where $\vx + a = (x_1 + a,\ldots,x_n+a)$. Namely, if we apply a translation to the inputs, the mechanism must output a location that is the result of the same translation to the original output.
\end{definition}

\begin{definition}[Dictatorship]
A mechanism $f$ is \emph{dictatorial} if and only if $\exists i \in N$, $\forall \vx$, $f(\vx) = x_i$. We say agent $i$ is the \emph{dictator}.
\end{definition}

\begin{definition}[2-Dictatorship]
A mechanism $f$ is \emph{2-dictatorial} if and only if $\exists i, j \in N$, $\forall \vx$, $f(\vx)$ lies on the segment between $x_i$ and $x_j$. We say agents $i$ and $j$ are the \emph{2-dictators}.
\end{definition}

We are also interested in designing a group-strategyproof mechanism while minimizing one of the following common objectives --- expected maximum cost or expected social cost, that is,
\[\mc(P, \vx) = \mathop{\E}_{y\sim P}\left[\max_{i\in N} \|x_i - y\|\right],\]
or \[\soc(P, \vx) = \E_{y\sim P} \left[\sum_{i\in N} \|x_i - y\| \right].\]
Also, we slightly abuse the notation, such that for $y\in \R_d$, we use $\mc(y,\vx)$ and $\soc(y, \vx)$ to denote the deterministic version of the objective functions. We say a mechanism is an $\alpha$-approximation of the optimum with respect to an objective $\obj$ if $\forall \vx \in \R_d^n$,
\[\obj(f(\vx), \vx) \leq \alpha \min_{y\in \R_d} \obj(y, \vx).\]

\section{Deterministic Mechanisms}
\label{section:deterministic}


The trivial dictatorship mechanism is group-strategyproof, $2$-approximation for maximum cost, and $(n - 1)$-approximation for social cost. Unfortunately, we will show that this is the only possibility of deterministic, unanimous, and group-strategyproof mechanisms in any strictly convex space.


We start with a lemma which quickly follows by the definition of unanimity.
\begin{lemma}\label{UnanimousInfApprox}
All bounded-approximation mechanisms are unanimous.
\end{lemma}
\begin{proof}
For any $y \in \R_d$, consider the case where $x_1=\cdots=x_n=y$.

For both maximum cost and social cost,
\[\mc(y, \vx) = \soc(y, \vx) = 0,\]
so it must be
\[\mc(f(\vx), \vx) = \soc(f(\vx), \vx) = 0.\]
By the definition of $\mc(\cdot)$ and $\soc(\cdot)$, $f(\vx) = y$, that is, $f$ is unanimous.
\end{proof}

In some papers, they use \emph{onto} as one of the constraints instead of \emph{unanimous} (e.g., \cite{Schummer2002Network}), which are known to be equivalent when $f$ is deterministic and strategyproof. Intuitively, given any output location, we can simply move every agent to the output location one by one and the output must stand still, otherwise it would contradict strategyproofness. Here we give a short proof.
\begin{proposition}\label{UnanimousOnto}
Suppose $f$ is deterministic and strategyproof, then $f$ is unanimous if and only if $f$ is onto.
\end{proposition}
\begin{proof}
If $f$ is unanimous, $\forall y$, $f(\vx) = y$ when $x_1=\cdots=x_n=y$, so $f$ is onto.

If $f$ is onto, then $\forall y$, there exists a profile $\vx$ such that $f(\vx) = y$. For all $i \in N$, when we move $x_i$ to $y$, because of strategyproofness, $\|f(y, \vx_{-i}) - y\| \leq \|f(\vx) - y\| = 0$, that is, $f(y, \vx_{-i}) = y$. After we move all agents to $y$ (one by one), the output still stays unchanged, and thus $f$ is unanimous.
\end{proof}

In the same manner, we obtain the uncompromising property, which is simple but useful.
The name of this property refers back to an earlier paper~\cite{Border1983Straightforward}, which is meant to have a similar purpose.

\begin{lemma}[Uncompromising property] \label{StillDictatorLemma}
Let $f$ be a strategyproof mechanism. Let $\vx$ be a profile such that the output is deterministic, and let $y = f(\vx)$. We claim that $f(\vx') = y$, if either $x_i' = x_i$ or $x_i' = y$ holds for all $i \in N$.
\end{lemma}
\begin{proof}
$\forall i \in N$, we claim that $f(y, \vx_{-i}) = y$, that is, if $x_i$ moves to the output, the output will not change. If not, agent $i$ would gain by misreporting $x_i$ instead of $y$. Formally, 
\[\|f(y, \vx_{-i}) - y\|\leq \|f(\vx) - y\| = 0.\]

Applying the claim for multiple times, we can move any group of agents to $y$ one by one, while the output stands still. Therefore, $f(\vx') = y$.
\end{proof} 

Derived from strategyproofness, we then find a general property that characterizes the relation between one's movement and the cost. We will show that the distance from an agent to the output cannot have a sudden change when an agent is moving slowly.
Intuitively, whenever there is a sudden distance change, the agent can always gain by misreporting the one with lower cost. In the proof, we first establish the 1-Lipschitz property, i.e., the distance change cannot be larger than one's movement. This property
also holds for randomized mechanisms, so we consider the expected distance.

\begin{lemma}[Cost continuity] \label{ContinuousLemma}
Let $f$ be a strategyproof mechanism. $\forall i \in N$, for any fixed $\vx_{-i} \in \R_d^{n-1}$, the distance between $x_i$ and the output of $(x_i, \vx_{-i})$ 
\[\mu(x_i) \triangleq \|f(x_i, \vx_{-i}) - x_i\|\]
is a continuous function. Moreover, $\forall x_i, x_i' \in \R_d$,
\[\|\mu(x_i) - \mu(x_i')\| \leq \|x_i - x_i'\|.\]
\end{lemma}

\begin{proof}
Assume for contradiction that  $\exists x_i, x_i' \in \R_d$ such that $\|\mu(x_i) - \mu(x_i')\| > \|x_i - x_i'\|$. Without loss of generality, assume $\mu(x_i) - \mu(x_i') > \|x_i - x_i'\|$. If agent $i$ misreports $x_i'$ instead of $x_i$, then
\begin{align*}
\|f(x_i', \vx_{-i}) - x_i\| ~&\leq \|f(x_i', \vx_{-i}) - x_i'\| + \|x_i - x_i'\| \\
    ~&= \mu(x_i') + \|x_i - x_i'\| \\
    ~&< \mu(x_i) \\
    ~&= \|f(x_i, \vx_{-i}) - x_i\|,
\end{align*}
which disobeys strategyproofness. Therefore, $\forall x_i, x_i'$, $\|\mu(x_i) - \mu(x_i')\| \leq \|x_i - x_i'\|$.

$\forall x_i$, let $x_i' \to x_i$, then we obtain the continuity of $\mu$ at all points.
\end{proof}

\subsection{Output Space Reduction}



Output space reduction plays an important role in our characterizations; this can be done by considering the situation when all agents collaborate to misreport the same location, and then the output must be that point as well by unanimity. When $n = 2$, due to the strict convexity, as long as the output does not lie on the segment between the two agents, there exists such location where the group-strategyproofness can be violated.

\begin{lemma} \label{SegmentLemma}
Suppose $f$ is deterministic, unanimous, and group-strategyproof. When $n = 2$, $\forall \vx$, $f(\vx)$ lies on $\overline{x_1x_2}$~\footnote{For any $x\in \R_d$, $y\in \R_d$, we denote $\overline{xy}$ as the segment between $x$ and $y$, that is, the set $\{\xi x + (1-\xi) y \mid \xi \in [0,1]\}$.}.
\end{lemma}
\begin{proof}

Suppose for the sake of contradiction that there exists $\vx$ such that $f(\vx)$ does not lie on $\overline{x_1x_2}$. Due to the strict convexity, $f(\vx) - x_1$ and $f(\vx) - x_2$ are not in the same direction and thus the following triangle inequality holds strictly:
\begin{align*}
\|x_1 - f(\vx)\| + \|f(\vx) - x_2\| > \|x_1 - x_2\|.
\end{align*}
However, $\forall y \in \overline{x_1x_2}$,
\begin{align*}
\|x_1 - y\| + \|y - x_2\| = \|x_1 - x_2\|,
\end{align*}
and thus we can easily find an $y$ such that $\|y - x_1\| < \|f(\vx) - x_1\|$ and $\|y - x_2\| < \|f(\vx) - x_2\|$, i.e., $y$ is a strictly better choice for both agents. If both agents misreport $y$, then the output must be $y$ by unanimity, contradicting group-strategyproofness.
\end{proof}




\subsection{Base Case}

Then we begin with the base case where $n = 2$. The result of this base case will be used afterwards for induction.
The following lemma states that if it is dictatorial for some profile, then it is dictatorial over all profiles. 

\begin{lemma} \label{PartialDictatorLemma}
Suppose $f$ is deterministic, unanimous, and group-strategyproof. When $n = 2$, if there exist $x_1 \neq x_2$ such that $f(\vx) = x_1$, then agent $1$ is the dictator in all profiles.
\end{lemma}

In the proof (see Appendix \ref{ProofPartialDictatorLemma}), we start with a weaker lemma, that is, by changing the location of the other agent, the dictator remains the same. The intuition behind this lemma is simply to move by steps, while preserving the only possible output to be the location of the dictator. Also, the route requires at least two dimensions, so this statement is not true for the one-dimensional setting (think that one can choose the leftmost as well as the rightmost point).

Now we are ready for the dictatorship in this special case.
In the proof, we actually show that the condition of Lemma \ref{PartialDictatorLemma} always holds with respect to some agent, otherwise it would contradict the uncompromising property.

\begin{lemma} \label{SpecialDictatorialTheorem}
When $n = 2$, if $f$ is deterministic, unanimous, and group-strategyproof, then $f$ is dictatorial.
\end{lemma}
\begin{proof}
Let $\vx$ be any profile such that $x_1 \neq x_2$, and $y = f(\vx)$. However, if $y \neq x_1$ and $y \neq x_2$, by the uncompromising property, we obtain $f(y, x_2) = y$. Therefore, by Lemma \ref{PartialDictatorLemma}, agent $1$ is the dictator in all profiles, which contradicts $y \neq x_1$.

Since either $y = x_1$ or $y = x_2$, $f$ is dictatorial (by Lemma \ref{PartialDictatorLemma}).
\end{proof}

\subsection{Scale Reduction}

Our theorem below generalizes the base case for any $n\geq 2$.
Intuitively, if we divide the agents into two groups where each group of agents shares the same location, we may conclude that one of them is the group of dictators. We first show that there exists a group of dictators containing $n-1$ agents, and then reduce the $n$-agent game into an $(n-1)$-agent problem by fixing the location of the non-dictator. By induction, there should be a dictator in the group, and finally we show that this dictator keeps to be the same agent regardless of the location of the non-dictator.

\begin{theorem} \label{DictatorialTheorem}
If $f$ is deterministic, unanimous, and group-strategyproof, then $f$ is dictatorial.
\end{theorem}
\begin{proof}
We prove by induction on $n$. By Lemma \ref{SpecialDictatorialTheorem}, it holds when $n = 2$.
Now we assume that $n \geq 3$, and assume it holds for $n - 1$.


First, we construct two new mechanisms $g_1$, $g_2$, where there are only two agents for each: $\forall \vy = (y_1, y_2)$, let
\begin{align*}
g_1(y_1, y_2) &= f(y_1, y_2, \dots, y_2), \\
g_2(y_1, y_2) &= f(y_1, \dots, y_1, y_2).
\end{align*}
In short, we bind the two groups of agents $\{2, \dots, n\}$, $\{1, \dots, n - 1\}$ respectively and then construct $g_1$, $g_2$.

It is clear that $g_1$, $g_2$ are deterministic, unanimous, and group-strategyproof, so $g_1$, $g_2$ are dictatorial (by Lemma \ref{SpecialDictatorialTheorem}). If agent $1$ ($y_1$) is the dictator of $g_1$, then agent $2$ ($y_2$) is \emph{not} the dictator of $g_2$, because, by the uncompromising property,
\[f(y_1, y_2, \dots, y_2) = y_1 \implies f(y_1, \dots, y_1, y_2) = y_1.\]
Thus, either agent $2$ is the dictator of $g_1$, or agent $1$ is the dictator of $g_2$. Without loss of generality, assume agent $2$ is the dictator of $g_1$.


Then we construct a set of mechanisms, where there are exactly $n - 1$ agents for each: $\forall x_1$, $\forall \vx_{-1}$, let
\[f_{x_1}(\vx_{-1}) = f(x_{1}, \vx_{-1}).\]
It is clear that each reduced mechanism $f_{x_1}$ is deterministic and group-strategyproof, while its unanimity comes from the dictatorship of $g_1$. Therefore, $\forall x_1$, $f_{x_1}$ is dictatorial (by the induction assumption).

Then it suffices to show that all $f_{x_1}$ have a common dictator. Suppose for contradiction that there exist $x_1, x_1'$ such that $f_{x_1}$, $f_{x_1'}$ have different dictators. Without loss of generality, assume agents $2, 3$ are the dictators of $f_{x_1}$, $f_{x_1'}$, respectively. Consider the following profiles:
\begin{align*}
\vx &= (x_1, x_1', x_1, \vx_{-\{1, 2, 3\}}), \\
\vx' &= (x_1', x_1', x_1, \vx_{-\{1, 2, 3\}}),
\end{align*}
where $\vx_{-\{1, 2, 3\}}$ can be arbitrary. In this case, agent $1$ would misreport $x_1'$ instead of truthfully reporting $x_1$, as $f(\vx) = x_1'$ and $f(\vx') = x_1$, which leads to a contradiction.

Therefore, there exists $i \in \{2, \dots, n\}$ such that $\forall x_1$, $f_{x_1}$ is dictatorial and agent $i$ is the dictator. That is, $f$ is dictatorial.
\end{proof}


Theorem \ref{DictatorialTheorem} also shows an impossibility result with respect to the \emph{anonymity} --- a commonly used property in the literature, when the voters are unwilling to be identified.

\begin{corollary} \label{DictatorialCorollary}
No deterministic, unanimous, group-strategyproof mechanism is anonymous.
\end{corollary}

\section{Randomized Mechanisms}
\label{section:randomized}


Compared with deterministic mechanisms, it seems that randomized mechanisms have more potential to achieve better approximations. We first consider the following mechanism, which is a variant of the prototype proposed by Procaccia and Tennenholtz~\cite{Procaccia2009Approximate}. In their one-dimensional setting, a $3/2$-approximation for maximum cost is guaranteed by randomly selecting the leftmost point, the rightmost point, and the midpoint. However, this is unachievable in multi-dimensional space and thus we simply use the locations of two fixed agents instead.

\begin{mechanism} \label{RandMed}
Given $\vx$, return $x_1$ with $1/4$ probability, $x_2$ with $1/4$ probability, and $(x_1 + x_2)/2$ with $1/2$ probability.
\end{mechanism}


When $n = 2$, Mechanism \ref{RandMed} is group-strategyproof and $3/2$-approximation for maximum cost (by a similar proof to~\cite{Procaccia2009Approximate}), but it reduces to a trivial $2$-approximation for any $n \geq 3$. For social cost, it ensures an $n/2$-approximation (see Appendix \ref{ProofRandMed} for the proof), but finally we will show that this is almost tight as well.

Under the deterministic setting, translational invariance is not used since dictatorship already implies translational invariance. However, after we remove the deterministic constraint, there exist some strange mechanisms which are not translation-invariant. 

\begin{mechanism} \label{Separate2Dictator}
Given $\vx$, let $r$ be the first coordinate of $x_1$, and $y$ be the point on $\overline{x_1x_2}$ such that $\|x_1-y\| = \min\{|r - a|,\|x_1-x_2\|\}$. Similarly, let $y'$ be the point on $\overline{x_1x_3}$ such that $\|x_1-y'\| = \min\{|r - a|,\|x_1-x_3\|\}$. 

If $r \geq a$, return $x_1$ with $2/3$ probability, and $y$ with $1/3$ probability. Otherwise, return $x_1$ with $2/3$ probability, and $y'$ with $1/3$ probability.
\end{mechanism}

\begin{figure}
    \centering
    \subfigure[]{
    \label{SepCase1}
    \includegraphics[scale=0.1]{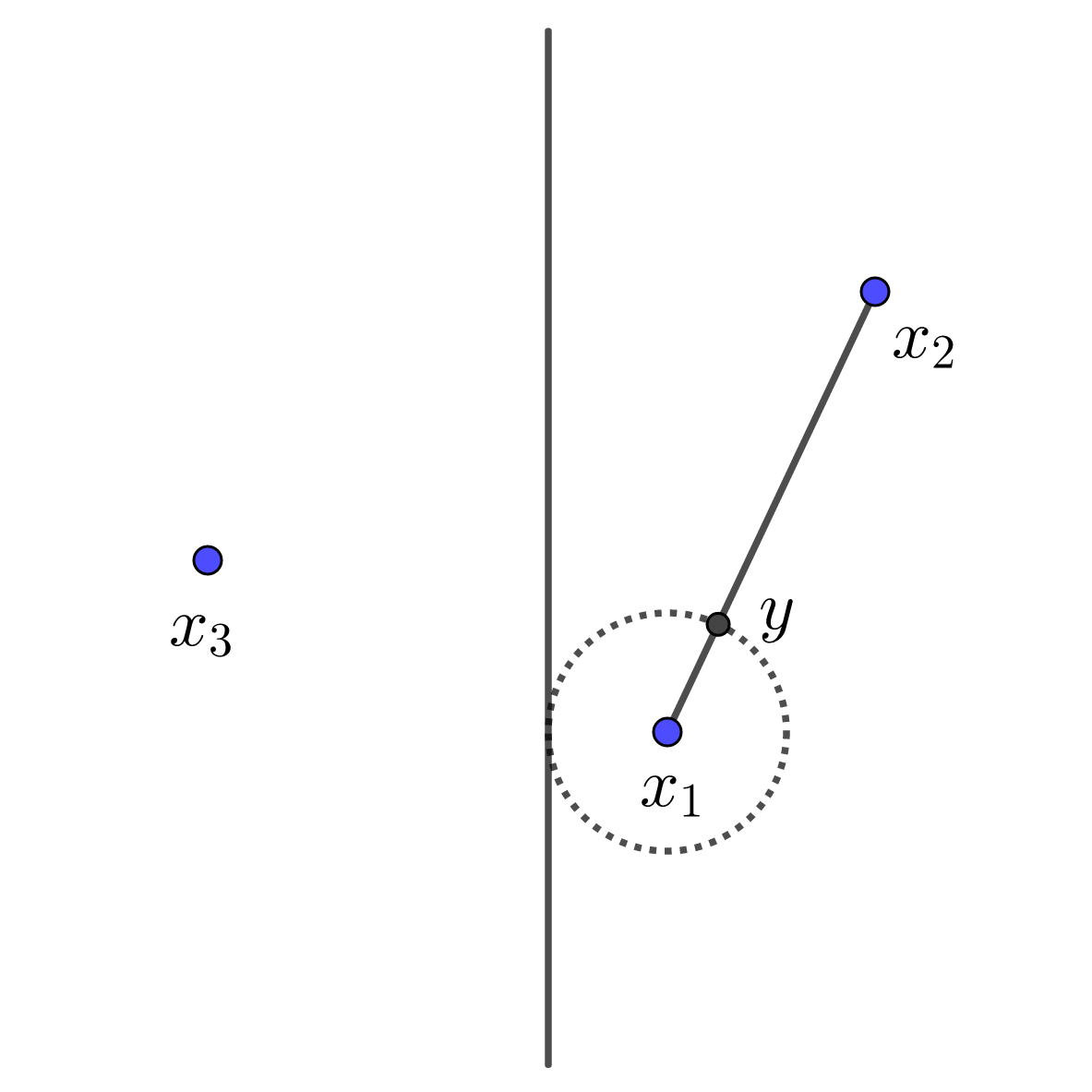}}
    \subfigure[]{
    \label{SepCase2}
    \includegraphics[scale=0.1]{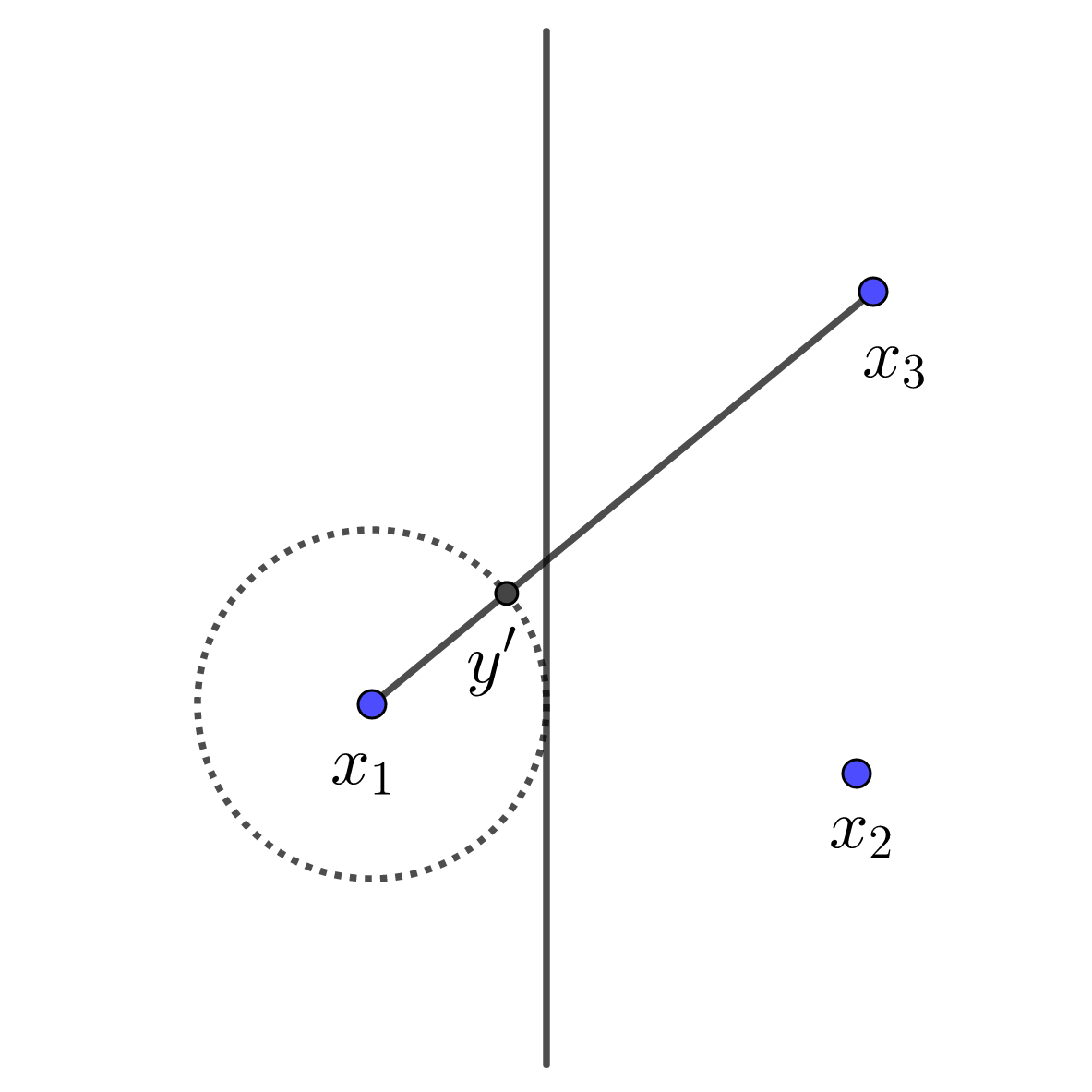}}
    \subfigure[]{
    \label{SepCase3}
    \includegraphics[scale=0.1]{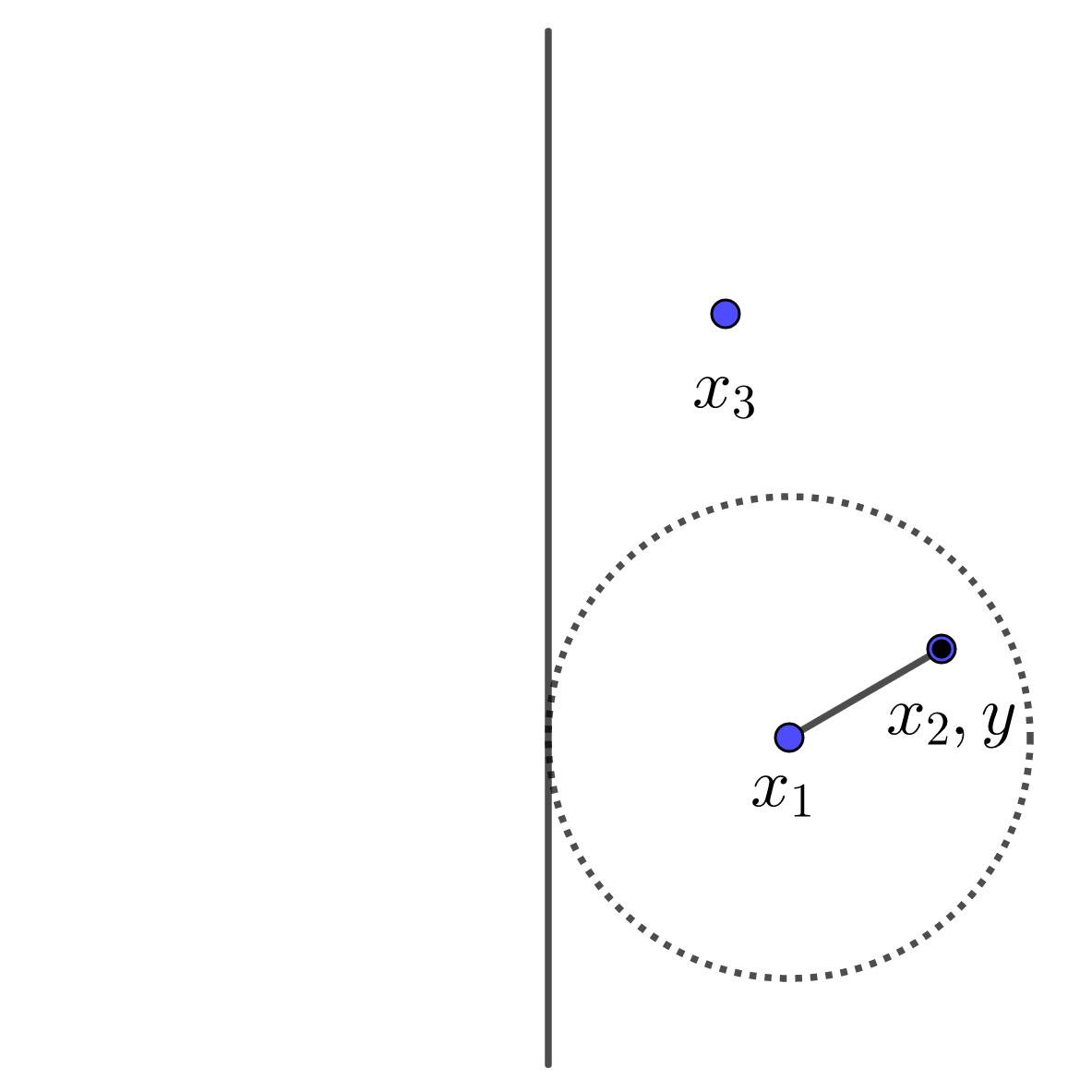}}
    \caption{Illustration of Mechanism \ref{Separate2Dictator}. Three possible cases are shown in the 2-dimensional Euclidean space. The vertical lines stand for the points of which the first coordinate equals to $a$, and the circles are of radius $|r-a|$.}
    \label{SepIllustration}
\end{figure}

For any fixed constant $a$, Mechanism \ref{Separate2Dictator} is unanimous and group-strategyproof (see Appendix \ref{ProofSeparate2Dictator} for the proof), but it is not translation-invariant. For example, supposing $r \geq a$ and $r - a< \|x_1-x_2\|$, when $r$ increases, $y$ gets more and more close to $x_2$. Intuitively, if a mechanism is not translation-invariant, it might be related to some constant (e.g., $a$ in Mechanism \ref{Separate2Dictator}), which is not a desirable property.

For better illustration, we use the two notations below --- \emph{centroid} and \emph{radius} of a distribution. Let the centroid of a distribution $P$ be \[\mathcal C(P) = \E_{y\sim P} y,\] and the radius of $P$ be \[\mathcal R(P) = \E_{y\sim P} \|y - \mathcal C(P)\| = \|P - \mathcal C(P)\|.\]



\subsection{Output Space Reduction}

Generally, the output can be a distribution over the whole space. Our first lemma in this section reduces every group-strategyproof and unanimous mechanism into distributions on the line segments, derived from the convexity at the centroid.
Consider the situation when all agents collaborate to misreport the centroid. 
By the convexity of the distance function, no agent would be worse off after misreporting, but strict preferences cannot hold in some cases. Specifically, if the support of the output forms a line and some agent lies on the same line on one side, then misreporting the centroid is indifferent to her. If there exists one such agent or multiple agents on the same side, then they can choose a slightly closer location which they strictly prefer.
 
\begin{lemma}\label{CentroidLemma}
Suppose $f$ is unanimous and group-strategyproof. $\forall \vx$, either $\mathcal R(f(\vx)) = 0$ (that is, the output is deterministic), or $\exists i, j \in N$ such that the support of $f(\vx)$ lies on the segment between $x_i$ and $x_j$.
\end{lemma}

\begin{proof}

Suppose $\mathcal R(f(\vx)) > 0$ (i.e., the output is strictly randomized) and let $y = \mathcal C(f(\vx))$. According to the group-strategyproofness, at least one agent cannot gain if all agents misreport the same location $y$. That is, by unanimity, $\exists i\in N$ such that 
\begin{equation}\label{eq1}
    \|f(\vx) - x_i\| \leq \|y - x_i\|.
\end{equation}
Let $N_1$ be the set of agents satisfying inequality (\ref{eq1}), and $N_2 = N\setminus N_1$.

Due to the strict convexity, $\forall i \in N$, $\|f(\vx)-x_i\| \geq \|y - x_i\|$, and moreover, $\|f(\vx)-x_i\| = \|y - x_i\|$ (i.e., $i \in N_1$) if and only if the support of $f(\vx)$ lies on a line with $x_i$ on the same side. Formally, $\|f(\vx)-x_i\| = \|y - x_i\|$ if and only if there exists a unit vector $e$ such that $\forall y \in \text{supp}(f(\vx))$, $\exists c \geq 0$, $y = x_i + c \cdot e$.


If all agents in $N_1$ share the same unit vector $e$ defined above, then we can find an $\epsilon$ where $\epsilon < \|f(\vx) - x_j\| - \|y - x_j\|$ for all $j\in N_2$, and $\epsilon < \|y - x_i\|$ for all $i\in N_1$. Let $y' = y - \epsilon \cdot e$. Then, $\forall j\in N_2$,
\[\|y' - x_j\| \leq \|y' - y\| + \|y - x_j\| = \epsilon + \|y- x_j\| < \|f(\vx) - x_j\|,\]
and $\forall i \in N_1$,
\[\|y' - x_j\| = \|y - x_j\| - \epsilon = \|f(\vx) - x_i\| - \epsilon < \|f(\vx) - x_i\|. \]
Thus, it violates group-strategyproofness if all agents collaborate to misreport $y'$.

Otherwise, there exist agents $i, j \in N_1$ with two opposite unit vectors $e$ and $-e$, which means that the support of $f(\vx)$ lies on the different sides to $x_i$ and $x_j$ on the same line, that is, the support of $f(\vx)$ lies on the segment between $x_i$ and $x_j$.
\end{proof}

Combined with Lemma \ref{SegmentLemma} and the cost continuity, we can obtain some easy characterizations of 2-agent games, as the following corollaries state.

\begin{corollary} \label{RandSegmentCorollary}
When $n = 2$, if $f$ is unanimous and group-strategyproof, then $\forall \vx$, $f(\vx)$ lies on $\overline{x_1 x_2}$.
\end{corollary}

\begin{corollary} \label{RandContinuousCorollary}
When $n = 2$, if $f$ is unanimous and group-strategyproof, then $\mathcal C(f(x_1, x_2))$ is a continuous function.
\end{corollary}


It is also meaningful to think of what would happen if \emph{strong} group-strategyproofness should be guaranteed, i.e., no group of agents can jointly misreport their preferences so that \emph{at least one} member is strictly better off without making any member worse off.
Following the idea of this lemma, strong group-strategyproofness could lead to a clear dictatorship even with randomness. If $f$ is supposed to be strong group-strategyproof instead, we can observe that (by a similar proof) in the case of $\mathcal R(f(\vx)) > 0$, $\vx$ must be on a line, and $f(\vx)$ is also on the same line, i.e., the output must be deterministic as long as the agents are not on a line.





%

\subsection{Base Case}

Lemma \ref{CentroidLemma} reduces every randomized output to a line between some agents, however, the lines are not necessarily formed by the same two agents and here our situation is still very complicated. In what follows, we aim to characterize the base case where $n = 3$, $f$ is unanimous and group-strategyproof. The following lemma shows that under some conditions, the output can be relatively very close to the input of some agent $j$ and isolated from some agent $i$. This is a pivot step of our final result.

\begin{lemma} \label{InfiniteRatioLemma}
When $n = 3$, if $f$ is unanimous and group-strategyproof, then $\forall z \in \R_d$, $\forall \ell > 0$, $\exists i \in N$,  $\forall \epsilon > 0$, there exists $\vx$ such that \[\|f(\vx) - x_j\| < \epsilon \|x_i - x_j\| < \ell\]
holds for some $j \neq i$ where $\|x_j - z\| < \ell$.
\end{lemma}

To prove this lemma, we construct a sequence of profiles such that each profile forms an isosceles triangle, while the output keeps on the base of the triangle. Moreover, we ensure that as the agents move, the legs of the triangle get longer, while the base gets shorter. After that, we can easily find a profile in the sequence such that $\|f(\vx) - x_j\| <\epsilon \|x_i - x_j\|$, where $x_i$ is the apex of the triangle.

To make the output remains on the base, a clever application of the cost continuity is needed. Intuitively, if the output jumps from the base to a leg, then the distance must not be continuous for some agents. However, the distance only needs to be kept continuous with respect to the moving agent, and this intuition should be further clarified. We designate two agents and try to move either one at a time. If the output jumps when one agent moves and also when the other agent moves, we show the contradiction by considering the situation where both agents move.
\begin{proof}[Proof of Lemma \ref{InfiniteRatioLemma}]
For any $z$ and $\ell$, consider an equilateral triangle $\Delta x_1x_2x_3$ with edges of length $\ell$, such that $\|x_i - z\|< \ell$ for $i\in N$. By Lemma \ref{CentroidLemma}, there are two cases: $f(x_1,x_2,x_3)$ is deterministic and is in the triangle; or there exist $i,j$, such that $f(x_1,x_2,x_3)$ lies on $\overline{x_ix_j}$. For the first case, let $y = f(x_1,x_2,x_3)$, without loss of generality, we assume $x_3 \neq y$. By the uncompromising property, $f(y,y,x_3) = y$. That is, the lemma is satisfied for $i = 3$. For the second case, without loss of generality, we assume $f(x_1,x_2,x_3)$ lies on $\overline{x_1x_2}$. Let $x_3$ be the ``$x_i$'' in the lemma. We will construct a sequence of location profiles $(x_1,x_2,x_3)$ in the plane defined by the equilateral triangle such that 
\begin{enumerate}[(a)]
    \item $\Delta x_1x_2x_3$ remains to be an isosceles triangle, with $x_3$ be the apex;
    \item in each iteration, $x_3$ becomes further away from $x_1$ and $x_2$, or $x_1$ and $x_2$ get closer, i.e., either $\|x_3 - x_1\| = \|x_3 - x_2\|$ increases, or $\|x_1 - x_2\|$ decreases;
    \item $\|x_1 - z\|< \ell$, $\|x_2 - z\|< \ell$;
    \item $\|x_3 - x_1\| = \|x_3 - x_2\| < \ell /\epsilon$.
    \item the output $f(\vx)$ always lies on $\overline{x_1x_2}$;
\end{enumerate}
Particularly, in each round, we will find $x_1',x_2', x_3'$, and inductively prove that  $\exists i\in N$, the location profile $(x_i',\vx_{-i})$ satisfies the properties above. 


For all $\delta_1, \delta_2, \delta_3$ (such that $\delta_1\leq \|x_1-x_2\|$ and $\delta_2\leq \|x_1-x_2\|$), we can easily find $x_1', x_2', x_3'$ such that  $\|x_1-x_1'\| = \delta_1$, $\|x_2 - x_2'\|=\delta_2, \|x_3 - x_3'\|=\delta_3$, and $\forall i \in N$ the location profile $(x_i',\vx_{-i})$ satisfies the property (a)-(c). For property (d), if  $\|x_3-x_1\| = \|x_3 - x_2\| > \ell/2\epsilon$, then we already have \[\min\{\|f(\vx) - x_1\|, \|f(\vx) - x_2\|\}\leq \ell /2 = \epsilon \cdot \ell/2\epsilon < \epsilon \|x_2-x_1\|.\] Therefore, we can assume $\|x_3 - x_1\| \leq \ell/2\epsilon$, namely, if we set $\delta_3 < \ell /2\epsilon$, property (d) is also satisfied.
For property (e), we first rule out the case where $f(x_i',\vx_{-i})$ is deterministic and strictly inside or outside the triangle, since in that case, we can move agents $1, 2$ to the output and the output remains the same by the uncompromising property, and thus the lemma is satisfied for any $\epsilon$. In what follows, we only consider the case that $f(x_i',\vx_{-i})$ lies on some edge, and we will show either one of $(x_i',\vx_{-i})$ satisfies property (e), or $f(\vx)$ is extremely close to $x_1$ or $x_2$.

Let $P_1 = f(x_1,x_2,x_3)$, $P_2 = f(x_1,x_2',x_3)$, $P_3 = f(x_1,x_2,x_3')$, and $P_4=f(x_1,x_2',x_3')$.  If $P_3$ lies on $\overline{x_1x_2}$, we are done. If not, without loss of generality, we assume that $P_3$ lies on $\overline{x_2x_3'}$. Now consider $P_2$, if $P_2$ lies on $\overline{x_1x_2'}$, we are done. Otherwise we consider the following two cases separately:

\begin{figure}[htbp]
    \centering
    \subfigure[]{
    \label{TriangleCase1}
    \includegraphics[scale=0.1]{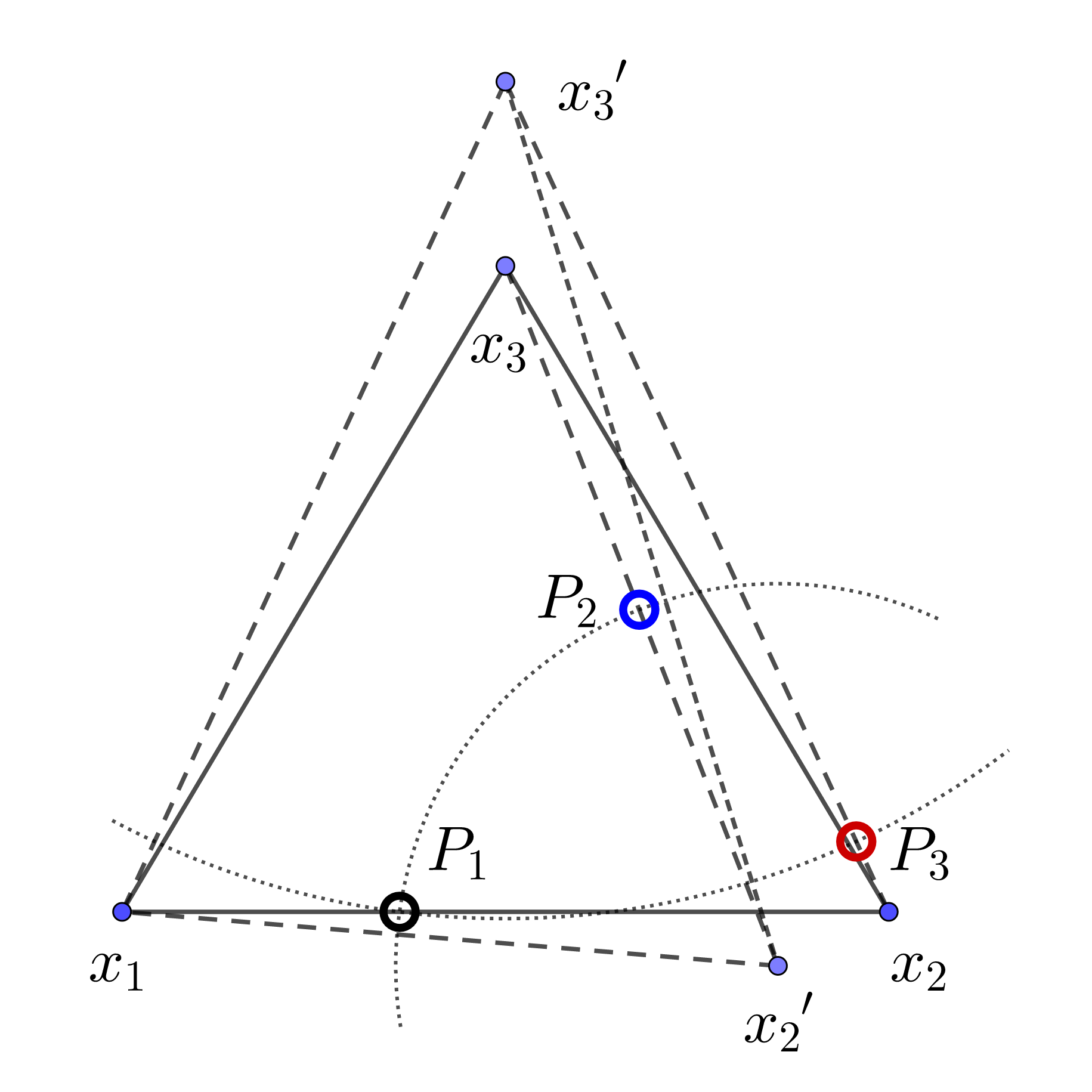}}
    \subfigure[]{
    \label{TriangleCase2}
    \includegraphics[scale=0.1]{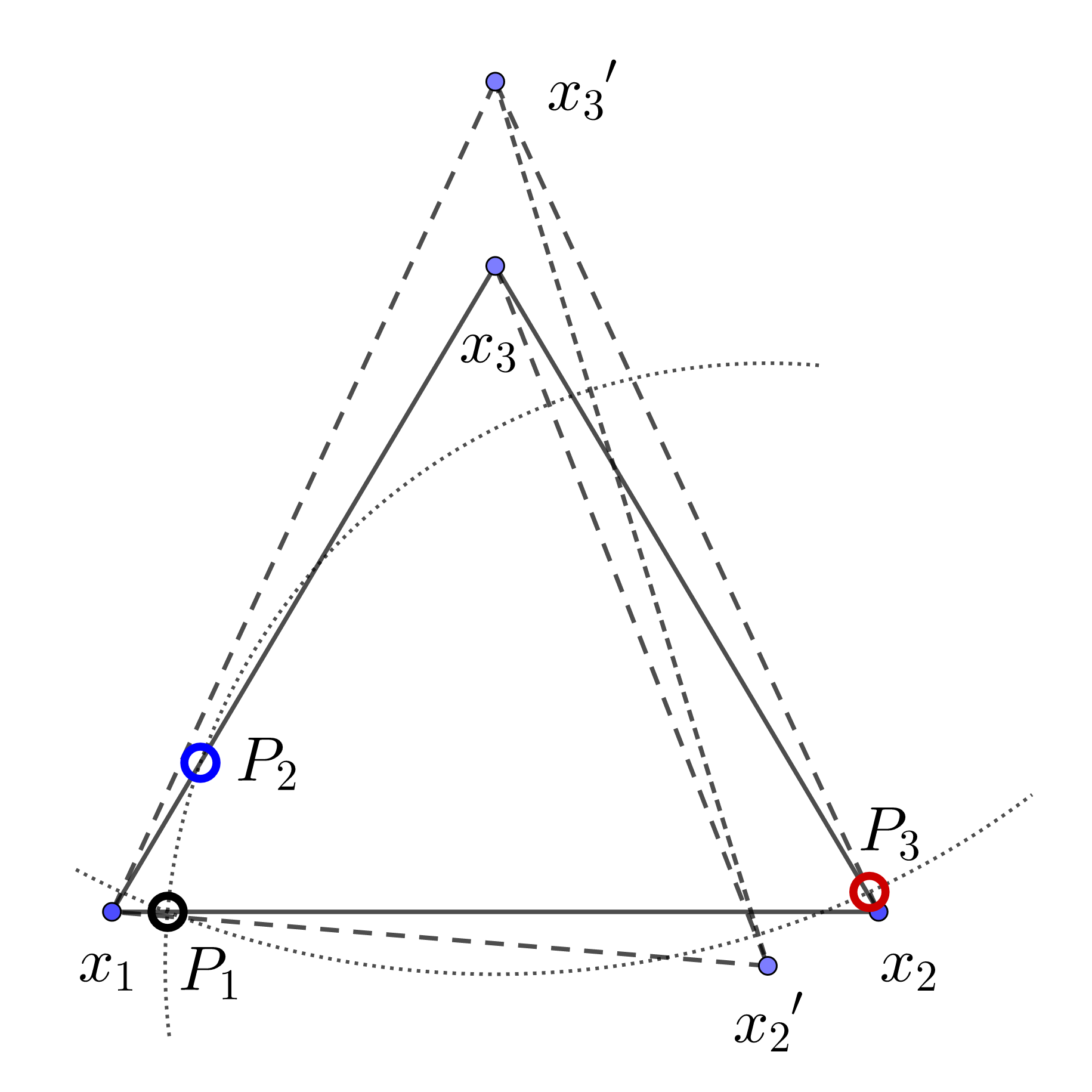}}
    \caption{Proof of Lemma \ref{InfiniteRatioLemma}. Two possible cases are shown in the Euclidean space, where neither $P_2$ nor $P_3$ stays on the same edge with $P_1$.}
\end{figure}

\begin{enumerate}
    \item $P_2$ lies on $\overline{x_2'x_3}$ (see Figure \ref{TriangleCase1}). Then we have 
    \begin{align*}
        \|x_3'-P_4\| \leq&~ \|x_3 - P_2\| + \delta_3 & \text{(cost continuity)}\\
        =&~ \|x_2' - x_3\| - \|x_2'-P_2\| + \delta_3 & \\
        \leq&~ \|x_2' - x_3\| - \|x_2 - P_1\| + \delta_2 + \delta_3. & \text{(cost continuity)}
    \end{align*}
    Similarly, 
    \begin{align*}
        \|x_2'-P_4\| \leq&~ \|x_2 - P_3\| + \delta_2 & \text{(cost continuity)}\\
        =&~ \|x_2 - x_3'\| - \|x_3'-P_3\| + \delta_2 & \\
        \leq&~ \|x_2 - x_3'\| - \|x_3 - P_1\| + \delta_2 + \delta_3. & \text{(cost continuity)}
    \end{align*}
    Adding two inequalities together, 
    \begin{align*}
        \|x_2'-x_3'\| \leq&~\|x_3'-P_4\|+\|x_2'-P_4\|\nonumber\\
        \leq&~ \|x_2' - x_3\| - \|x_2 - P_1\| + \|x_2 - x_3'\| - \|x_3 - P_1\| + 2\delta_2 + 2\delta_3\nonumber\\
        \leq&~ \|x_2'-x_3'\| + \|x_2 - x_3\| - (\|x_2 - P_1\| + \|x_3 - P_1\|) + 3\delta_2 + 3\delta_3\label{eq2}.
    \end{align*}
    Let $C_1 = \mathcal C(P_1)$, then by convexity,
    \begin{equation}
        3\delta_2 + 3\delta_3\geq \|x_2 - P_1\|+\|x_3-P_1\| - \|x_2 - x_3\|\geq \|x_2 - C_1\|+\|x_3-C_1\| - \|x_2 - x_3\|.\label{eq:case1}
    \end{equation}

    \item $P_2$ lies on $\overline{x_1x_3}$ (see Figure \ref{TriangleCase2}). First by the cost continuity, 
    \begin{equation}
        \|x_2-P_2\| \leq \|x_2' - P_2\| + \delta_2 \leq \|x_2 - P_1\| + 2\delta_2.\label{eq2}
    \end{equation}
    Let $C_2=\mathcal C(P_2)$, then
    \[\|x_1-P_2\| + \|P_2 - x_2\|\geq \|x_1-C_2\| + \|C_2 - x_2\| \geq \|x_1 - x_2\| = \|x_1 - P_1\| + \|x_2 - P_1\|.\]
    Combined with (\ref{eq2}), we have
    \[\|x_1 - P_1\|\leq \|x_1 - P_2\| + 2\delta_2.\]
    
    Also, by cost continuity,
    \begin{align*}
        \|x_3'-P_4\| \leq&~ \|x_3 - P_2\| + \delta_3.
    \end{align*}
    On the other hand,
    \begin{align*}
        \|x_2' - P_4\| \leq&~ \|x_2 - P_3\| + \delta_2 &\text{(cost continuity)}\\
        =&~ \|x_2-x_3'\| - \|x_3' - P_3\| + \delta_2 & \\
        \leq&~ \|x_2 - x_3'\| - \|x_3 - P_1\| + \delta_3 + \delta_2. &\text{(cost continuity)}
    \end{align*}
    Adding the three inequalities above, 
    \begin{align*}
        \|x_2' - x_3'\| =&~ \|x_2' - P_4\| + \|x_3' - P_4\|\\
        \leq&~ \|x_2-x_3'\| - \|x_3 - P_1\| + \|x_3 - P_2\| + \|x_1 - P_2\| - \|x_1 - P_1\| + 3\delta_2 + 2\delta_3\\
        =&~ \|x_2-x_3'\| - \|x_3 - P_1\| + \|x_3 - x_1\| - \|x_1 - P_1\| + 3\delta_2 + 2\delta_3.
    \end{align*}
    Let $C_1 = \mathcal C(P_1)$, then by convexity,
    \begin{equation}
        3\delta_2 + 2\delta_2 \geq \|x_1 - P_1\| + \|x_3 - P_1\| - \|x_3 - x_1\| \geq \|x_1 - C_1\| + \|x_3 - C_1\| - \|x_3 - x_1\|.\label{eq:case2}
    \end{equation}
\end{enumerate}

One can observe the right hand sides of both (\ref{eq:case1}) and (\ref{eq:case2}) are in a form of 
\[\|x_k - C_1\| + \|x_3 - C_1\| - \|x_k - x_3\|, \]
where $k=1, 2$. By thinking the formula as a function of $\|x_3 - C_1\|$, $\|x_k - x_3\|$, $\angle x_3x_kC_1$ and the slope of $\overline{x_kx_4}$, it is easy to see that for any $\alpha$, there exists $\delta' > 0$, such that when $\|x_k - x_3\|\leq \ell/2\epsilon$, $\|x_k -C_1\|\geq \epsilon\ell/2$ and $\angle x_3x_kC_1\geq \alpha$,  $\|x_k - C_1\| + \|x_3 - C_1\| - \|x_k - x_3\|$ is always lower bounded by $\delta'$.
In that case, if we set both $\delta_2$ and $\delta_3$ to be $\delta'/10$, the two possible cases are rejected, that is,  either $P_2$ lies on $\overline{x_1x_2'}$, or $P_3$ lies on $\overline{x_1x_2}$.

Note $\|x_k - x_3\|\leq \ell/2\epsilon$, $\angle x_3x_kC_1\geq \alpha$ are always satisfied if $\alpha$ is set to be $\min\{\angle x_3x_1x_2, \angle x_3x_2x_1\}$ at the beginning (because $\angle x_3x_kC_1$ always increases). 
It suffices to consider the case $\|x_k -C_1\|< \epsilon\ell/2$. 
Let $r = \min\{\|x_2 - C_1\|, \|x_1 - C_1\|\}$. Then \[\mathcal R(P_1) = \E_{y\sim P_1} \|y - C_1\|\leq 2r.\] If we move $x_1, x_2$ to $C_1$, i.e., consider $\vx' = (C_1, C_1, x_3)$, then by group-strategyproofness we have \[\|f(\vx') - C_1\|\leq \|f(\vx) - C_1\| = \|P_1 - C_1\| = \mathcal R(P_1) \leq 2r < \epsilon\ell \leq \epsilon\|x_3 - x_1\|.\]
That is, we have already found the profile $\vx'$ satisfying the requirement of the lemma.
\end{proof}
In the following lemma, from a set of convergent profiles obtained above, we prove that as long as some two agents share the same input, the output must be exactly their input. Intuitively, starting from an appropriate profile, we can then move the isolated agent to anywhere we want by steps, while keeping the output relatively very close to the fixed agent.




\begin{lemma} \label{ConvergeLemma}
Consider the case where $n = 3$, $f$ is unanimous, translation-invariant, and group-strategyproof. Then $\exists i \in N$, $\forall z \in \R_d$, $\forall \vx$ such that $\vx_{-i} = (z,\ldots,z)$, $f(\vx) = z$.
\end{lemma}

Much care should be taken when moving randomized distributions. Essentially, we show that the relative distance between the output and an agent may be amplified but bounded by a constant, which actually decomposes into two more smaller lemmas in our proof (see Appendix \ref{ProofConvergeLemma}).

Note that this property just likes the unanimity if we disregard the isolated agent $i$. We can then complete the special case where $n=3$, following a similar proof to Lemma \ref{SegmentLemma}.

\begin{lemma} \label{Special2DictatorialTheorem}
When $n = 3$, if $f$ is unanimous, translation-invariant, and group-strategyproof, then $f$ is 2-dictatorial.
\end{lemma}
\begin{proof}

Let $i$ be the agent chosen in Lemma \ref{ConvergeLemma}, and let $j, k$ be the other two agents. We claim that agents $j, k$ are the 2-dictators.

Assume for contradiction that there exists $\vx$ such that $f(\vx)$ does not lie on $\overline{x_j x_k}$. Due to the strict convexity, there exists $y \in \R_d$ such that $\|y - x_j\| < \|f(\vx) - x_j\|$ and $\|y - x_k\| < \|f(\vx) - x_k\|$. Let $\vx'$ be the profile where $x_i' = x_i$ and $x_j' = x_k' = y$. By Lemma \ref{ConvergeLemma}, $f(\vx') = y$, so agents $j$ and $k$ would collaborate to misreport $y$, contradicting group-strategyproofness.
\end{proof}

\subsection{Scale Reduction}

Theorem \ref{2DictatorialTheorem} generalizes the base case for any $n \geq 3$. Similarly, we can divide the agents into three groups, among which two will be the groups of 2-dictators. However, here we cannot simply fix the location of an agent as done in Theorem \ref{DictatorialTheorem}, because then the reduced game would not be translation-invariant. Differently, by fixing the location of a non-dictator to another agent (i.e., bind them together), we can then construct an $(n-1)$-agent game which is still unanimous, translation-invariant, and group-strategyproof, and thus 2-dictatorial by the induction assumption. Leveraging some ``partial unanimity'' properties, we finally show that the 2-dictatorship in the reduced mechanism extends to all profiles no matter how the non-dictator locates.

\begin{theorem} \label{2DictatorialTheorem}
If $f$ is unanimous, translation-invariant, and group-strategyproof, then $f$ is 2-dictatorial.
\end{theorem}

\begin{proof}
We prove by induction on $n$. When $n = 2$, by Corollary \ref{RandSegmentCorollary}, the only two agents are the 2-dictators. Now we assume that $n \geq 3$, and assume it holds for $n - 1$.


First, we divide all agents into three non-empty, nonintersecting groups $S_1, S_2, S_3$, where  $S_1 \cup S_2 \cup S_3 = N$. $\forall y_1, y_2, y_3 \in \R_d$, denote $\vx(y_1, y_2, y_3)$ as the profile where $x_i = y_1$ for all $i \in S_1$, $x_j = y_2$ for all $j \in S_2$, and $x_k = y_3$ for all $k \in S_3$. We construct a mechanism $g$, where there are only three agents: $\forall \vy = (y_1, y_2, y_3)$, let
\begin{align*}
g(\vy) &= f(\vx(y_1, y_2, y_3)).
\end{align*}
In short, we bind the three groups of agents respectively and then construct $g$. It is clear that $g$ is unanimous, translation-invariant and group-strategyproof, so $g$ is 2-dictatorial (by Lemma \ref{Special2DictatorialTheorem}). Without loss of generality, assume agents $2, 3$ are the 2-dictators of $g$, and $1 \in S_1$.

Let $\mathcal{F}_S$ denote the set of mechanisms derived from $f$ when the locations of all agents in $S$ are fixed. Formally, $\forall S \subseteq N$,
\[\mathcal{F}_S = \{f_{\vx_S} \mid \vx_S \in \R_d^{|S|}\},\]
where $\forall \vx_{-S}$, $f_{\vx_S}(\vx_{-S}) = f(\vx_S, \vx_{-S})$.
Note that every mechanism in $\mathcal{F}_S$ is formed by $n - |S|$ agents, and is group-strategyproof as well as $f$. Similarly let $\mathcal{F}_{-S}$ denote the set of mechanisms when the locations of all agents \emph{not} in $S$ are fixed. Specially, $\mathcal{F}_1$ denotes the set of mechanisms when the location of agent $1$ is fixed, and moreover, every mechanism in $\mathcal{F}_1$ (as well as $\mathcal{F}_{S_1}$) is unanimous (by the 2-dictatorship of $g$ together with the uncompromising property).

Let $f_1$ denote the mechanism derived from $f$ when the location of agent $1$ is fixed to the location of agent $2$, i.e., $\forall \vx_{-1}$,
\[f_1(\vx_{-1}) = f(x_2, \vx_{-1}).\]
It is clear that $f_1$ is unanimous, translation-invariant, and group-strategyproof, and thus 2-dictatorial by the induction assumption. We assume that agents $1, 2$ (i.e., agents $2, 3$ in $f$) are the 2-dictators of $f_1$. This assumption is without loss of generality since otherwise we can bind agent $1$ to one of the 2-dictators instead of agent $2$, while the 2-dictators must not change by the uncompromising property. By the 2-dictatorship of $f_1$, as long as agents $1, 2, 3$ in $f$ share the same location, the output must be that point as well, i.e., every mechanism in $\mathcal{F}_{-\{1,2,3\}}$ is unanimous.

Then we will generalize the 2-dictatorship of $f_1$ to the whole space for arbitrary $x_1$. Assume for contradiction that $f$ is not 2-dictatorial, i.e., $\exists \hat{\vx}$ such that $f(\hat{\vx})$ does not lie on $\overline{x_2 x_3}$. Consider the following two cases:
\begin{itemize}
\item $\exists \hat{\vx}$ such that $f(\hat{\vx})$ does not lie on $\overline{x_2 x_3}$ and $\mathcal{R}(f(\hat{\vx})) > 0$. By Lemma \ref{CentroidLemma}, $f(\hat{\vx})$ lies on a line segment between some agents. By the same argument as the proof of Lemma \ref{CentroidLemma}, there exists $x' \in \R_d$ on the same line with $f(\hat{\vx})$ such that $\|x' - x_i\| < \|f(\hat{\vx}) - x_i\|$ for all $i \in S$, where $S$ includes agents $2, 3$ and all agents that are not on the same line with $f(\hat{\vx})$. Let $\hat{\vx}'$ denote the profile after moving all agents in $S$ to $x'$ (then agents $2, \dots, n$ are on the same line with $f(\hat{\vx})$), which satisfies $f(\hat{\vx}') \neq x'$ by group-strategyproofness. Since every mechanism in $\mathcal{F}_1$ and $\mathcal{F}_{-\{1,2,3\}}$ is unanimous and group-strategyproof, by Lemma \ref{CentroidLemma}, $f(\hat{\vx}')$ must both lie on the same line with $f(\hat{\vx})$ and lie on $\overline{x'x_1}$. If $x_1$ is not on the same line with $f(\hat{\vx})$, then it directly follows that $f(\hat{\vx}') = x'$ and leads to a contradiction. Even if $x_1$ is on the same line with $f(\hat{\vx})$, any $f(\hat{\vx}') \neq x'$ would violate the continuity of the centroid (Corollary \ref{RandContinuousCorollary}), when considering the two-player game formed by agent $1$ and a group of agents $2, 3$ while keeping the other agents fixed.

\item $\forall \vx$ such that $f(\vx)$ does not lie on $\overline{x_2x_3}$, $\mathcal{R}(f(\vx)) = 0$. We construct a two-player game formed by agent $1$ and a group of agents $2, 3$ while keeping the other agents fixed: $\forall x_1, x_2 \in \R_d$, let
\[\hat{f}(x_1, x_2) = f(x_1, x_2, x_2, \hat{\vx}_{-\{1,2,3\}}).\]
It is clear that $\hat{f}$ is deterministic, unanimous, and group-strategyproof, and thus dictatorial by Theorem \ref{DictatorialTheorem}. That is, as long as agents $2, 3$ share the same location, the output must be that point as well regardless of the location of agent $1$. Now consider the two-player game $f_{\hat{\vx}_{-\{2,3\}}}$ formed by agents $2, 3$: $\forall x_2, x_3 \in \R_d$,
\[f_{\hat{\vx}_{-\{2,3\}}}(x_2, x_3) = f(\hat{\vx}_1, x_2, x_3, \hat{\vx}_{-\{1,2,3\}}).\]
Since $f_{\hat{\vx}_{-\{2,3\}}}$ is unanimous and group-strategyproof, $f(\hat{\vx})$ must lie on $\overline{x_2x_3}$ by Corollary \ref{RandSegmentCorollary}, which makes a contradiction.

\end{itemize}

Therefore, $f$ is 2-dictatorial.

\end{proof}

\begin{corollary} \label{2DictatorialCorollary}
No unanimous, translation-invariant, group-strategyproof mechanism is anonymous, for any $n \geq 3$.
\end{corollary}

Theorem \ref{2DictatorialTheorem} indicates that we can simply move the 2-dictators to one side away from the others, where the lower bounds of approximately optimal mechanisms are obtained.

\begin{corollary}
For any $n \geq 3$, no translation-invariant, group-strategyproof mechanism can do better than $2$-approximation for maximum cost.
\end{corollary}

\begin{corollary}
For any $n \geq 3$, no translation-invariant, group-strategyproof mechanism can do better than $\left(n/2 - 1\right)$-approximation for social cost.
\end{corollary}

\section{Discussion and Open Problems} \label{sec:discussion}

Our characterization of 2-dictatorship is almost complete, given the existence of non-dictatorial mechanisms (e.g.,  Mechanism~\ref{RandMed}). However, one remaining problem is whether or not the distribution between the 2-dictators could be affected by the other agents. The answer of this question would close the gap for the social cost objective.

In this paper, we focus on the strictly convex norms, which rule out $L_1$ and $L_\infty$ space.
Previous work shows that generalized median voter schemes are strategyproof under $L_1$-norm~\cite{Barbera1993Median}, and we suggest that they can moreover guarantee group-strategyproofness in a 2-dimensional $L_1$ (or $L_\infty$, after rotating the axes) space.
However, they are no longer group-strategyproof if there are more than 2 dimensions. For example, suppose there are 5 agents located in a 3-dimensional $L_1$ space with coordinates $(1, 0, 0), (0, 1, 0), (0, 0, 1), (1, 1, 1), (1, 1, 1)$, and the algorithm selects a median in each dimension, making the output to be $(1, 1, 1)$; then the first three agents can collaborate to misreport $(0, 0, 0)$, which would result in a strictly better output $(0, 0, 0)$. 

The most challenging generalization of this work would be deriving characterization of  randomized strategyproof mechanisms in multi-dimensional space. Among deterministic mechanisms, much effort has been done to the generalized median voter schemes under different domains; given such characterization, we suggest that selecting the median in each dimension can be approximately efficient for social cost, but it also indicates that $2$-approximation is already the tight bound for maximum cost. Yet there is no known result about randomized mechanisms, even in an approximation view. We propose the following mechanism for better approximating the maximum cost.

\begin{mechanism} \label{RandCenter}
Given $\vx$, output $(x_1 + \dots + x_n) / n$ with $1/2$ probability, and each $x_i$ with $1/2n$ probability.
\end{mechanism}  

\noindent
Mechanism \ref{RandCenter} is strategyproof and $(2 - 1/n)$-approximation for maximum cost in any normed vector space (see Appendix \ref{ProofRandCenter} for the proof). Although $2 - o(1)$ may not be a significant breakthrough, it breaks the tight bound of deterministic mechanisms and provides positive implications. Further, the characterization of randomized mechanisms would be completely different if taking away the group-strategyproofness constraint, since some randomized strategyproof mechanisms such as Mechanism \ref{RandCenter} can truly break the limit studied in this paper. Studying randomized strategyproof mechanisms would be a more challenging but attractive task. We believe that our approach and techniques could potentially lead the way for future algorithmic studies and further characterizations in this fundamental setting.



\section*{Acknowledgements}
P. Tang and S. Zhao are supported by Science and Technology Innovation 2030 --- ``New Generation Artificial Intelligence'' Major Project No. 2018AAA0100904, 2018AAA0101103 and Beijing Academy of Artificial Intelligence (BAAI).
D. Yu is supported by NSF, ONR, Simons Foundation, Schmidt Foundation, Amazon Research, DARPA and SRC.

\bibliographystyle{ACM-Reference-Format}
\bibliography{ref}

\clearpage

\begin{appendices}

\section{Deterministic Mechanisms}\label{DeterAppendix}

\subsection{Proof of Lemma \ref{PartialDictatorLemma}}\label{ProofPartialDictatorLemma}

\begin{lemma} \label{WeakDictatorLemma}
Suppose $f$ is deterministic, unanimous, and group-strategyproof. When $n = 2$, if there exist $x_1 \neq x_2$ such that $f(\vx) = x_1$, then for all $x_2'$, $f(x_1, x_2') = x_1$.
\end{lemma}

\begin{proof}

If $x_2' = x_1$, by unanimity, $f(x_1, x_2') = x_1$.

Otherwise, by Lemma \ref{SegmentLemma}, $f(x_1, x_2')$ lies on the segment between $x_1$ and $x_2'$. Thus if $\|x_2' - x_2\| < \|x_1 - x_2\|$, then $\|f(x_1,x'_2) - x_2\| \geq \|f(x_1, x_2) - x_2\|$ if and only if $f(x_1, x'_2) = x_1$, and we can obtain that $f(x_1, x'_2) = x_1$ by strategyproofness.

Now we only need to consider the case that $\|x_2' - x_2\| \geq \|x_1 - x_2\|$. In this case, 
it is actually easy to construct a sequence of points which starts with $x_2$, ends with $x_2'$, and for any two consecutive points $y$ and $z$, $\|z - y\| < \|x_1 - y\|$. E.g., first move agent $2$ from $x_2$ to the ray directed from $x_1$ to $x_2'$ while keeping the radius, then move straight to $x_2'$ (see Figure \ref{route_lemma}). Then we can inductively prove that for any point $y$ in the sequence, $f(x_1,y)=x_1$. Finally, we obtain $f(x_1, x_2') = x_1$.%
%
%
%
\begin{figure}[htbp]%
\centering
\includegraphics[scale=0.1]{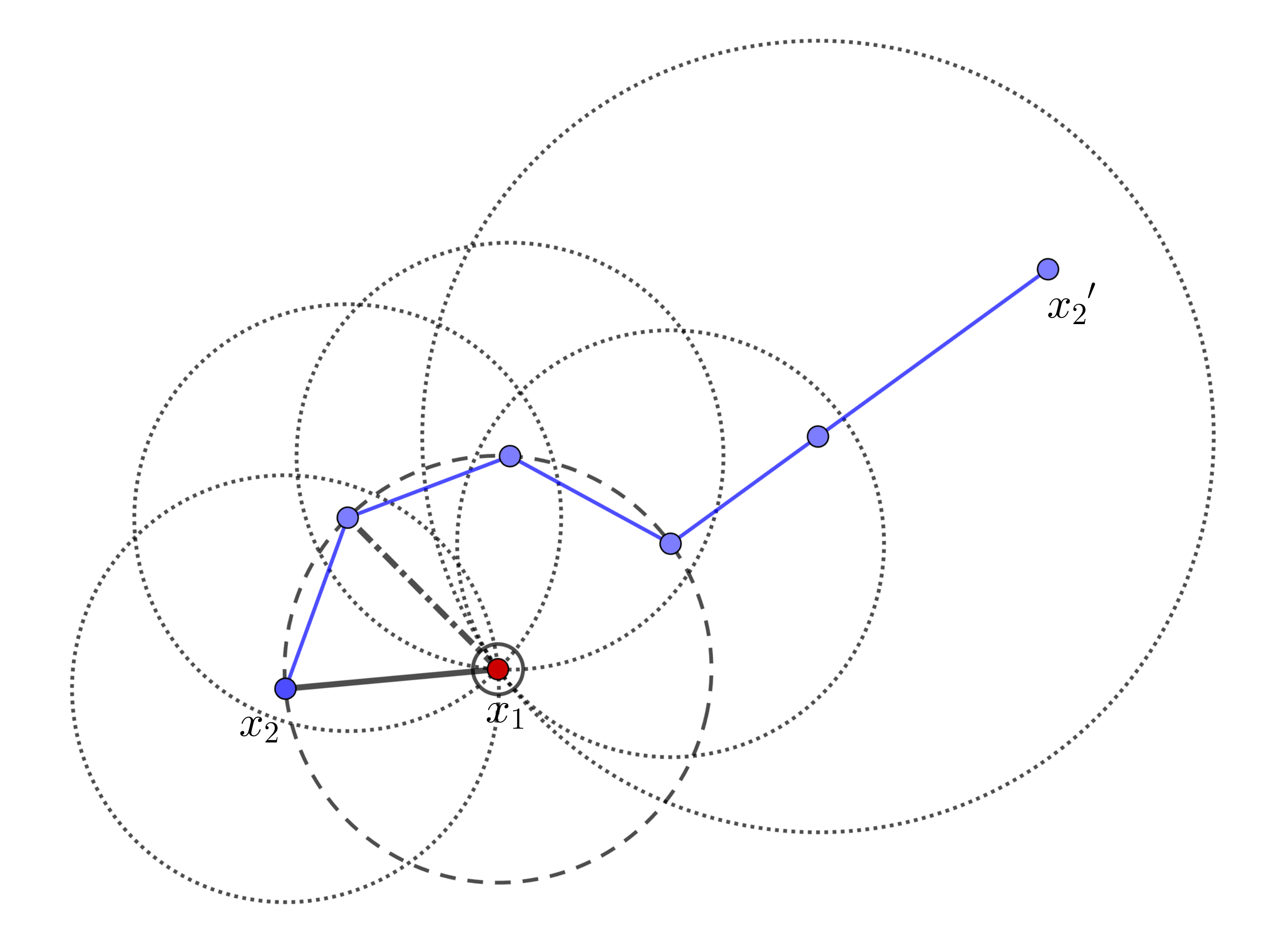}
\caption{Proof of Lemma \ref{WeakDictatorLemma}, shown in the Euclidean space. It shows a possible route that starts from $x_2$ and ends at $x_2'$, while keeping the dictator agent $1$ staying at the center $x_1$.}
\label{route_lemma}
\end{figure}%
\end{proof}

\begin{replemma}{PartialDictatorLemma}
Suppose $f$ is deterministic, unanimous, and group-strategyproof. When $n = 2$, if there exist $x_1 \neq x_2$ such that $f(\vx) = x_1$, then agent $1$ is the dictator in all profiles.
\end{replemma}
\begin{proof}

$\forall \vx' = (x_1', x_2')$, we prove $f(\vx') = x_1'$ by the following three steps.

\begin{enumerate}
\item Move agent $2$ to $\hat{x}_2$, such that $\hat{x}_2 \neq x_1'$ and $x_1'$ is on $\overline{x_1\hat{x}_2}$. By Lemma \ref{WeakDictatorLemma}, $f(x_1, \hat{x}_2) = x_1$.
\item Move agent $1$ to $x_1'$.

    Suppose $f(x_1', \hat{x}_2) \neq x_1'$. $\forall \xi \in [0, 1]$, let
    \[\mu(\xi) = \|f(\xi x_1 + (1 - \xi) x_1', \hat{x}_2) - \hat{x}_2\|.\]
    Also we know that $f(x_1', \hat{x}_2)$ lies on $\overline{x_1'\hat{x}_2}$ by Lemma \ref{SegmentLemma}, so \[\mu(0) = \|f(x_1', \hat{x}_2) - \hat{x}_2\| < \|x_1' - \hat{x}_2\| \leq \|x_1 - \hat{x}_2\| =  \mu(1).\] Because $\mu$ is a continuous function (by Lemma \ref{ContinuousLemma}), there exists $\xi_0\in(0,1]$ such that $\mu(\xi_0) = \|x_1' - \hat{x}_2\|$. Let $y = \xi_0 x_1 + (1 - \xi_0) x_1'$. We have $\|f(y, \hat{x}_2) - \hat{x}_2\| = \|x_1' - \hat{x}_2\|$, and $y$ is on  $\overline{x_1'\hat{x}_2}$ (by Lemma \ref{SegmentLemma}), so $f(y, \hat{x}_2) = x_1'$. Consequently, agent $1$ would misreport $y$ instead of $x_1'$, contradicting strategyproofness. Thus, $f(x_1', \hat{x}_2) = x_1'$.
\item Move agent $2$ from $\hat{x}_2$ to $x_2'$. By Lemma \ref{WeakDictatorLemma}, $f(\vx') = x_1'$.\qedhere
\end{enumerate}
\end{proof}





\section{Randomized Mechanisms}\label{RandomizeAppendix}

\subsection{Proof of Lemma \ref{ConvergeLemma}} \label{ProofConvergeLemma}

In order to prove Lemma \ref{ConvergeLemma}, we introduce Lemma \ref{DeltaBoundLemma} and Lemma \ref{DeltaUnanimousLemma}.

\begin{lemma} \label{DeltaBoundLemma}
Suppose $f$ is unanimous and group-strategyproof. When $n = 2$, $\forall \delta \geq 0$, if there exist $x_1 \neq x_2$ such that $\|f(\vx) - x_1\| \leq \delta$, then for all $x_2'$ where $\|x_2' - x_2\| <\|x_2 - x_1\|$,
\[\|f(x_1, x_2') - x_1\| \leq \frac{\delta}{1 - \frac{\|x_2' - x_2\|}{\|x_2 - x_1\|}}.\]
\end{lemma}
\begin{proof}

Let $r = \|x_2 - x_1\|$, $d = \|x_2' - x_2\|$. By Corollary \ref{RandSegmentCorollary}, $f(x_1, x_2')$ lies on $\overline{x_1 x_2'}$. $\forall x \in [0, r]$, let
\[\mu(x) = \left\|\frac{x}{r}x_1 + \left(1 - \frac{x}{r}\right) x_2' - x_2\right\|.\]
Clearly, $\mu(0) = r$, $\mu(r) = d$. By the convexity of distance, $\forall x \in [0, r]$,
\[\mu(x) \leq \left(1 - \frac{x}{r}\right) r + \frac{x}{r}\cdot d = r - \left(1 - \frac{d}{r}\right) x.\]
Thus, by strategyproofness,
\begin{align*}
r - \delta &\leq \|f(x_1, x_2) - x_2\| \\
    &\leq \|f(x_1, x_2') - x_2\|\\
    &= \mu(\|f(x_1, x_2') - x_1\|) \\
    &\leq r - \left(1 - \frac{d}{r}\right)\|f(x_1, x_2') - x_1\|,
\end{align*}
and we can solve that
\[\|f(x_1, x_2') - x_1\| \leq \frac{\delta}{1 - \frac{d}{r}}.\]

\end{proof}

\begin{lemma} \label{DeltaUnanimousLemma}
Suppose $f$ is unanimous and group-strategyproof. When $n = 2$, $\forall \delta \geq 0$, if there exist $x_1 \neq x_2$ such that $\|f(\vx) - x_1\| \leq \delta$, then for all $x_2'$ where $\ell \geq r$,
\[\|f(x_1, x_2') - x_1\| \leq \frac{100 \ell}{r} \delta.\]
Here, $r = \|x_2 - x_1\|$, $\ell = \|x_2' - x_1\|$.
\end{lemma}
\begin{proof}
The proof consists of two steps.
\begin{enumerate}
\item If $r = \ell$, then $\|f(x_1, x_2') - x_1\| \leq 50 \delta$.

We can move agent $2$ from $x_2$ to $x_2'$ in several steps, while keeping the radius (i.e., distance to $x_1$) unchanged. Note that we are actually moving along a convex contour. If we move a distance of at most $r/8$ at each step, then it can arrive within $32$ steps (think of moving along a square), and by Lemma \ref{DeltaBoundLemma},
\[\|f(x_1, x_2') - x_1\| \leq \left(\frac{8}{7}\right)^{32} \delta \leq 50 \delta.\]

\item If $x_2$ lies on $\overline{x_1x_2'}$, then $\|f(x_1, x_2') - x_1\| \leq 2 \ell \delta / r$.

From $x_2$ straight to $x_2'$, assume each time we move a distance of $\epsilon$ times the distance between $x_1$ and the current location of agent $2$. As $\epsilon \to 0$, by Lemma \ref{DeltaBoundLemma},
\begin{align*}
\|f(x_1, x_2') - x_1\| \leq \delta (1 - \epsilon)^{-1 - \log_{1 + \epsilon}{\frac{\ell}{r}}} \to \frac{\ell}{r} \delta,
\end{align*}
so there exists $\epsilon$ such that $\|f(x_1, x_2') - x_1\| \leq 2 \ell \delta / r$.
\end{enumerate}
We can then obtain the statement for all $x_2'$, by first applying step 1 and then applying step 2 (following a similar route to Figure \ref{route_lemma}).
\end{proof}

\begin{replemma}{ConvergeLemma}
Consider the case where $n = 3$, $f$ is unanimous, translation-invariant, and group-strategyproof. Then $\exists i \in N$, $\forall z \in \R_d$, $\forall \vx$ such that $\vx_{-i} = (z,\ldots,z)$, $f(\vx) = z$.
\end{replemma}
\begin{proof}

Let $\chi: \R_d \times \R \to N$ be a function such that $\chi(z, \ell)$ corresponds to the agent $i$ chosen by Lemma \ref{InfiniteRatioLemma} with respect to the given $z$ and $\ell$. Let $\Phi: \R_d \to 2^N$ be a function mapping a location to a subset of $N$, representing the convergence of $\chi(z, \ell)$ as $\ell \to 0$. Formally, $\forall z \in \R_d$, $\forall i \in N$, $i \in \Phi(z)$ if and only if either there is an infinite number of $m \in \mathbb{N}^+$ such that $\chi(z, 1 / m) = i$, or there exists $x_i \neq z$ such that $f(x_i, \vx_{-i}) = z$ where $\vx_{-i} = (z,\ldots,z)$. $\forall i \in N$, let $S_i = \{z \mid i \in \Phi(z)\}$. Clearly, $\Phi(z) \neq \varnothing$, and $S_1 \cup S_2 \cup S_3 = \R_d$. Specially, if $f$ is translation-invariant, then $\chi(z, l)$ and $\Phi(z)$ are invariant to $z$, so $S_i$ is either an empty set or exactly $\R_d$. Now we suppose that $S_i = \R_d$.

If $x_i = z$, then it holds by unanimity.

Otherwise, suppose $f(\vx) \neq z$. Let $\delta = \min(\|f(\vx) - z\|, 1) > 0$. By definition, $\ell$ can be infinitely small while satisfying $\chi(z, \ell) = i$. Let $\epsilon = \delta / 200$. Applying Lemma \ref{InfiniteRatioLemma} with respect to any $\ell < \epsilon \delta / 2$ where $\chi(z, \ell) = i$, there exists $\hat{\vx}$ such that
\[\|f(\hat{\vx}) - \hat{x}_j\| < \epsilon \|\hat{x}_i - \hat{x}_j\| < \frac{\epsilon \delta}{2}\]
holds for some $j \neq i$ where $\|\hat{x}_j - z\| < \epsilon \delta / 2$.

Let $z' = \hat{x}_j$. Now we move the agent other than $i, j$ to $z'$ (coincides with agent $j$), and these conditions still hold (by strategyproofness). Let $\vx'$ be the profile where $x_i' = x_i$ and $\vx_{-i}' = (z',\ldots,z')$. Because $\|\hat{x}_i - \hat{x}_j\| < \delta / 2 < \|x_i - z'\|$, we can apply Lemma \ref{DeltaUnanimousLemma} by considering agent $i$ and the others as a two-player game, so $\|f(\vx') - z'\| \leq \delta / 2$,
\begin{align*}
\|f(\vx') - z\| \leq \|f(\vx') - z'\| + \|z' - z\| < \delta.
\end{align*}
It leads to a contradiction, as agent $i$ would misreport $z'$ instead of $z$.
\end{proof}





\section{Proofs of Mechanisms}

\subsection{Strategyproofness and Approximation Bound of Mechanism \ref{RandCenter}} \label{ProofRandCenter}

\begin{repmechanism}{RandCenter}
Given $\vx$, output $(x_1 + \dots + x_n) / n$ with $1/2$ probability, and each $x_i$ with $1/2n$ probability.
\end{repmechanism}
\begin{proposition}
Mechanism \ref{RandCenter} is strategyproof.
\end{proposition}
\begin{proof}

Since all agents are symmetric, it suffices to show that agent $1$ cannot gain from misreporting. $\forall \vx$, $\forall x_1' \in \R_d$,
\begin{align*}
\|f(x_1', \vx_{-1}) - x_1\| &= \frac{1}{2}\left\|\frac{x_1' + \sum_{i=2}^{n}{x_i}}{n} - x_1\right\| + \frac{1}{2n}\|x_1' - x_1\| + \frac{1}{2n}\sum_{i=2}^{n}{\|x_i - x_1\|} \\
    &\geq \frac{1}{2}\left\|\frac{\sum_{i=1}^{n}{x_i}}{n} - x_1\right\| + \frac{1}{2n}\sum_{i=2}^{n}{\|x_i - x_1\|} \\
    &= \|f(\vx) - x_1\|.
\end{align*}
Agent $1$ cannot gain from misreporting, and thus Mechanism \ref{RandCenter} is strategyproof.

\end{proof}

\begin{proposition}
Mechanism \ref{RandCenter} is $(2 - 1/n)$-approximation for maximum cost.
\end{proposition}
\begin{proof}

For any $y_1 \neq y_2$, consider the following profile
\[\vx^* = (y_1, y_2, \dots, y_2).\]
In this case,
\[\mc(f(\vx^*), \vx^*) = \frac{n-1}{2n}\|y_1 - y_2\| + \frac{1}{2}\|y_1 - y_2\|  = \Big(1 - \frac{1}{2n}\Big)\|y_1 - y_2\|,\]
while the optimal is $\mc((y_1 + y_2)/2, \vx^*) = \|y_1 - y_2\|/2$, so the approximation ratio is at least $2 - 1/n$.

Now we prove it to be the upper bound. Assume for contradiction that $\exists \vx \in \R_d^n$, $\exists y \in \R_d$ such that
\[\frac{\mc(f(\vx), \vx)}{\mc(y, \vx)} > 2 - \frac{1}{n}.\]
Let $r = \mc(y, \vx)$. Let $\overline{x} = (x_1 + \dots + x_n)/n$. Then,
\[\Big(2 - \frac{1}{n}\Big) r < \mc(f(\vx), \vx) < \frac{1}{2}\mc(\overline{x}, \vx) + 1,\]
so there exists $i \in N$ such that 
\[\|x_i - \overline{x}\| = \mc(\overline{x}, \vx) > \Big(2 - \frac{2}{n}\Big)r.\]
Also we have
\[\|x_i - \overline{x}\| = \Bigg\|\frac{1}{n}\sum_{j \neq i}{x_i - x_j}\Bigg\| \leq \frac{1}{n}\sum_{j \neq i}{\|x_i - x_j\|},\]
so there exists $j \in N$ such that
\[\|x_i - x_j\| \geq \frac{n}{n-1}\|x_i - \overline{x}\| > 2r.\]
On the other hand,
\[\|x_i - x_j\| \leq \|x_i - y\| + \|x_j - y\| \leq 2r,\]
which makes a contradiction.

\end{proof}

\subsection{Approximation Bound of Mechanism \ref{RandMed}} \label{ProofRandMed}
\begin{repmechanism}{RandMed}
Given $\vx$, return $x_1$ with $1/4$ probability, $x_2$ with $1/4$ probability, and $(x_1 + x_2)/2$ with $1/2$ probability.
\end{repmechanism}
\begin{proposition}
Mechanism \ref{RandMed} is $n/2$-approximation for social cost.
\end{proposition}
\begin{proof}
For any $y_1 \neq y_2$, consider the following profile
\[\vx^* = (y_1, y_2, \dots, y_2).\]
Because $\soc(f(\vx^*), \vx^*) = \frac{n}{2}\|y_1 - y_2\|$ and the optimal is $\soc(y_1, \vx^*) = \|y_1 - y_2\|$, the approximation ratio is at least $n/2$.

Now we prove it to be the upper bound. For any $y \in \R_d$, assume $y$ is the optimal location that minimizes the social cost. Let $d_1 = \|x_1 - y\|$, $d_2 = \|x_2 - y\|$,
\[C = \sum_{i=3}^{n}{\|x_i - y\|}.\]
Due to the convexity, $\forall \vx$,
\[\soc(f(\vx), \vx) \leq \frac{\soc(x_1, \vx) + \soc(x_2, \vx)}{2},\]
and
\begin{align*}
\frac{\soc(f(\vx), \vx)}{\soc(y, \vx)} &\leq \frac{\frac{1}{2}\soc(x_1, \vx) + \frac{1}{2}\soc(x_2, \vx)}{\soc(y, \vx)} \\
    &\leq \frac{\frac{1}{2}((n - 2)d_1 + \|x_1 - x_2\| + C) + \frac{1}{2}((n - 2)d_2 + \|x_1 - x_2\| + C)}{d_1 + d_2 + C} \\
    &\leq \frac{\frac{n}{2} (d_1 + d_2) + C}{d_1 + d_2 + C} \\
    &\leq \frac{n}{2}.
\end{align*}
As this inequality holds for all $y \in \R_d$, we conclude that Mechanism \ref{RandMed} is exactly $n/2$-approximation for social cost.
\end{proof}

\subsection{Group-Strategyproofness of Mechanism \ref{Separate2Dictator}} \label{ProofSeparate2Dictator}
\begin{repmechanism}{Separate2Dictator}
Given $\vx$, let $r$ be the first coordinate of $x_1$, and $y$ be the point on $\overline{x_1x_2}$ such that $\|x_1-y\| = \min\{|r - a|,\|x_1-x_2\|\}$. Similarly, let $y'$ be the point on $\overline{x_1x_3}$ such that $\|x_1-y'\| = \min\{|r - a|,\|x_1-x_3\|\}$. 

If $r \geq a$, return $x_1$ with $2/3$ probability, and $y$ with $1/3$ probability. Otherwise, return $x_1$ with $2/3$ probability, and $y'$ with $1/3$ probability.
\end{repmechanism}
\begin{proposition}
Mechanism \ref{Separate2Dictator} is unanimous and group-strategyproof.
\end{proposition}
\begin{proof}

It suffices to show the group-strategyproofness for all $r \geq a$, as the other case is symmetric.

In this case, the output lies on $\overline{x_1 x_2}$, and determined by $x_1$ and $x_2$. Thus, a group that contains neither agent 1 nor agent 2 cannot violate group-strategyproofness. On the other hand, a group that contains both of agents $1, 2$ cannot violate group-strategyproofness either, because $\|x_1 - f(\vx)\| + \|x_2 - f(\vx)\| = \|x_1 - x_2\|$ has already reached the minimum. 

Consider a group that contains agent 2 but not agent 1. If $\|x_2 - x_1\| < r - a$, then $y = x_2$, so agent 2 is truthful. Otherwise, let $y'$ be the ``$y$'' after misreporting, where $\|y' - x_1\| \leq r - a$. As a result,  $\|x_2 - y'\| \geq \|x_2 - x_1\| - (r-a)= \|x_2 - y\|$. Therefore, this group cannot violate group-strategyproofness.

Consider a group that contains agent 1 but not agent 2. Let $y'$ be the ``$y$'' after agent 1 misreports $x_1'$. Let $d = \|x_1 - x_1'\|$, and $r'$ be the first coordinate of $x_1'$. Then we have $r'-a \geq r - a - d$ and $\|y' - x_1'\|\geq \|y - x_1\| - d$. Therefore, $\|y'-x_1\| \geq \|y - x_1\| - 2 d$, and
\begin{align*}
\|f(\vx') - x_1\| &= \frac{2}{3}\|x_1' - x_1\| + \frac{1}{3}\|y' - x_1\| \\
    &\geq \frac{2}{3}d +  \frac{1}{3}(\|y - x_1\| - 2d) \\
    &= \frac{\|y - x_1\|}{3} \\
    &= \|f(\vx) - x_1\|.
\end{align*}
Thus, this group cannot violate group-strategyproofness either.

The unanimity is clear, and thus Mechanism \ref{Separate2Dictator} is unanimous and group-strategyproof. 

\end{proof}

\end{appendices}

\appendix

\end{document}